\documentclass[draftclsnofoot,onecolumn,letterpaper]{IEEEtran}

\usepackage[utf8]{inputenc}
\usepackage[T1]{fontenc}
\usepackage{url}
\usepackage{ifthen}
\usepackage{cite}
\usepackage{enumitem}
\usepackage{comment}
\usepackage{graphicx}
\usepackage{multirow}
\usepackage{makecell}

\usepackage[caption=false,font=footnotesize]{subfig}

\usepackage[english]{babel}
\usepackage[cmex10]{amsmath}
\usepackage{amssymb}

\usepackage{mathrsfs}
\usepackage{tikz}
\usepackage[utf8]{inputenc}
\usepackage{amsfonts} 
\usepackage{amsthm}
\usepackage{epstopdf}
\usepackage{bbm}
\usepackage{enumerate}
\usepackage{enumitem}
\usepackage{psfrag,xspace}

\usepackage{bm}
\usepackage{hhline}
\usepackage{mathabx}
\usepackage{braket}
\usepackage{xcolor}
\definecolor{darkblue}{rgb}{0.1,0.1,0.8}
\definecolor{brickred}{rgb}{0.8, 0.25, 0.33}
\definecolor{DarkGreen}{rgb}{0,0.6,0}

\newtheorem{theorem}{Theorem}

\newtheorem{lemma}[theorem]{Lemma}

\newtheorem{corollary}[theorem]{Corollary}

\newtheorem{definition}[theorem]{Definition}

\newtheorem{remark}[theorem]{Remark}

\allowdisplaybreaks 
\makeatletter
\def\old@comma{,}
\catcode`\,=13
\def,{%
  \ifmmode%
    \old@comma\discretionary{}{}{}%
  \else%
    \old@comma%
  \fi%
}
\makeatother

\usepackage{etoolbox}
\providetoggle{Blue_revision}
\settoggle{Blue_revision}{true}

 \usepackage{hyperref}
\hypersetup{
colorlinks=true, %
 pdfstartview={FitH},
    linkcolor=red,
    citecolor=blue, 
    urlcolor={blue!80!black}
}

\usepackage{graphicx}
\usepackage{stmaryrd }
\usepackage{stackengine}

\usepackage{relsize}
\usepackage{upgreek}

\newcommand{\mH}{\mathsf{H}}

\newcommand{\mD}{\mathcal{D}}

\newcommand{\mP}{\mathcal{P}}

\newcommand{\mL}{\mathcal{L}}
\newcommand{\mX}{\mathcal{X}}

\newcommand{\mS}{\mathcal{S}}
\newcommand{\mN}{\mathcal{N}}
\newcommand{\mT}{\mathcal{T}}
\newcommand{\mM}{\mathcal{M}}

\DeclareMathOperator{\tr}{\mathrm{Tr}}
\newcommand{\Eutr}{\mathbb{E}_U\mathrm{Tr}}
\newcommand{\supp}{\mathrm{supp}}
\newcommand{\eig}{\mathrm{eig}}

\newcommand{\hsig}{\hat{\sigma}}
\newcommand{\hrho}{\hat{\rho}}

\newcommand{\htau}{\hat{\tau}}
\newcommand{\hmT}{\hat{\mT}}

\newcommand{\tlD}{\tilde{D}}
\newcommand{\tlQ}{\tilde{Q}}

\newcommand{\tlH}{\tilde{H}}

\newcommand{\tlmT}{\tilde{\mathcal{T}}}
\newcommand{\tlsig}{\tilde{\sigma}}
\newcommand{\tlrho}{\tilde{\rho}}
\newcommand{\tltau}{\tilde{\tau}}

\newcommand{\tenn}{{\otimes n}}

\newcommand{\id}{\mathrm{id}}

\newcommand{\de}{\mathrm{d}}
\newcommand{\rp}{\mathrm{rp}}

\newcommand{\da}{\partial_\alpha}
\newcommand{\tHa}{\tlH_\alpha^\uparrow(R|B)}
\newcommand{\Ha}{H_\alpha}

\newcommand{\eigrho}{\eig(\rho)_{\min}}

\begin{document}
\title{Quantum Soft Covering with Relative Entropy Criterion\vspace{2ex}} 

\author{%
 \IEEEauthorblockN{Xingyi He and S. Sandeep Pradhan}\\
 \IEEEauthorblockA{Department of Electrical Engineering and Computer Science, University of Michigan, USA\\
                   Email: \{xingyihe,  pradhanv\}@umich.edu}

\thanks{A preliminary version of this work was presented in part at the 2024 IEEE International Symposium on Information Theory (ISIT), 7-12 July 2024, Athens, Greece, doi:10.1109/ISIT57864.2024.10619111. \cite{he2024}}
\thanks{This work was supported in part by NSF grant CCF-2132815.}}

\thispagestyle{empty} 

\maketitle

\vspace{-1.5\baselineskip}

\begin{abstract}
   In this work, we propose a soft covering problem for fully quantum channels using relative entropy as a criterion for operator closeness. We establish covering lemmas by deriving one-shot bounds on the achievable rates in terms of smooth min-entropies. In the asymptotic regime, we show that the infimum of the rate, defined as the logarithm of the minimum rank of the encoded input state, is given by the minimal coherent information between the reference and output systems that yields the target output state. Furthermore, we present a one-shot quantum decoupling theorem that also employs a relative-entropy criterion. Due to the Pinsker inequality, our one-shot results based on the relative-entropy criterion are tighter than the corresponding results based on the trace norm considered in the literature. In addition, we establish achievable error exponents and second-order rates for quantum soft covering under both trace-distance and relative-entropy criteria.
\end{abstract}

\section{Introduction}

The quantum soft covering problem, also known as the channel output simulation problem, can be stated as follows. Suppose we are given a quantum channel and an objective density operator, we aim to generate an encoded state whose output approximates the target density operator. In the classical setting, a density operator reduces to a probabilistic distribution, and such a covering problem has been broadly investigated in the literature. The earliest work dates back to Wyner in his seminal paper \cite{wyner1975common}, where the problem was formulated as determining the rate of shared randomness required for two agents to generate classical i.i.d. samples of local random variables. The optimal rate of initial shared randomness is now known as Wyner’s common information. In subsequent works 
\cite{han2002approximation,hayashi2006general,bloch2013strong,cuff2009communication,cuff2013distributed,cuff2016soft}, the essence of soft covering was abstracted from this multiuser perspective and developed into a simulation problem based on a sphere-covering viewpoint. It is well known that the minimal achievable rate is given by the smallest mutual information between the channel input and output corresponding to the input distribution that generates the target output distribution. In the quantum setting, relevant discussions are generally classified into two genres based on the type of channel: classical-quantum (CQ) channels and fully quantum channels. Soft covering for CQ channels was first introduced in \cite{ahlswede2002strong}, where it was shown that the minimal achievable rate is characterized by replacing the classical mutual information by Holevo information, which serves exactly as the CQ analogue of the former. This result was later investigated and applied in many related works \cite{ahlswede2002strong,winter2004extrinsic,devetak2003classical,winter2005secret,devetak2005distillation,cai2004quantum,bennett2014quantum,luo2009channel,radhakrishnan2017one,shen2024optimal,hayashi2025resolvability,takahashi2026classical} \cite[Ch.~17]{wildebook}. 
The notion of soft covering for fully quantum channels was proposed recently \cite{atif2024quantum}, where the corresponding rate threshold is given by the coherent information. This highlights the duality connection with the fully quantum channel coding result. While the former surprisingly admits a single-letter expression, no such characterization exists for the latter.
Meanwhile, for quantum-classical (QC) channels, a related problem, known as measurement simulation, has also been studied in the literature \cite{winter2004extrinsic}\cite{2001compression,massar2000amount,wilde2012information,anshu2019convex}. 
Moreover, beyond the soft covering setting, a more general problem known as channel resolvability has also been introduced and studied extensively \cite{han2002approximation}\cite{hayashi2006general}\cite{hayashi2025resolvability}\cite{hayashi2012quantumwiretap,hayashi2006formula,watanabe2014strong} \cite[Sec.~9.4]{hayashi-book}. It is worth noting that channel resolvability is closely related to soft covering in the sense that the former deals with simulating all output states, including non-product ones, while the latter focuses solely on product ones. 

In the soft covering problem, it is important to specify a criterion of operator closeness to quantify how far the simulated state is from the target. A measure of closeness that has been applied in most of the literature is the trace distance (Schatten-1 norm). However, a stronger criterion that is arguably more information-theoretic is the relative entropy, which is tighter than the trace norm due to the Pinsker inequality. Furthermore, relative entropy enjoys certain useful properties such as recursivity \cite{csiszar2008axiomatic} and infinite dynamic range that may make it more attractive in certain applications, such as channel resolvability, secret communication with more nuanced discrimination between density operators.
In fact, the first work on the soft covering problem in the classical setting by Wyner \cite{wyner1975common}, considered the normalized relative entropy under an $n$-shot transmission. This was generalized to the unnormalized case in a subsequent work \cite{hayashi2006general}. Tighter criteria for classical channels, such as $f$-divergence or R\'enyi divergence \cite{liu2016e_,yu2018wyner,yu2018renyi,yu2019simulation,yu2025renyi,li2025two} have also been used to quantify the simulation error. Under these measures, the achievable rate may shift from the mutual information to other relevant quantities. 
For CQ channels, soft covering under relative-entropy error continues to have the Holevo information as the achievable rate infimum, and one-shot bounds have also been established \cite{hayashi2012quantumwiretap}. For fully quantum channels, the recently introduced soft covering problem \cite{atif2024quantum} still uses the trace-distance criterion, while the tighter covering behavior under the relative-entropy criterion remains unresolved.

In this work, we answer this question and formulate the soft covering problems for fully quantum channels with relative entropy as a measure of operator closeness. We characterize the achievable rate infima, i.e., the minimum rank of the encoded state to realize one-shot and asymptotic coverings, respectively. Specifically, in the one-shot case, the minimum rank is characterized by the smooth min-entropy between the reference and output systems (Theorem \ref{thm-cover-1shot}); in the asymptotic regime, this reduces to the single-letter coherent information between these two systems owing to the asymptotic equipartition property (AEP) of the smooth R\'enyi entropy (Theorem \ref{thm-cover-asym}). This is identical to the trace-distance soft covering result \cite[Thm.~3.5]{atif2024quantum}, demonstrating that when switching from the looser trace-distance criterion to the tighter relative-entropy criterion, the coherent information asymptotically remains sufficient.
Toward obtaining these results, we first propose a continuity bound (Lemma \ref{lem-contin-rela}) and a quadratic upper bound (Lemma \ref{lemma-D<Q2}) on relative entropy for subnormalized states. The proof of the one-shot result then proceeds by applying smoothing techniques \cite{renner2008security,tomamichel2010duality,dupuis2014one,dupuis-thesis,berta2009single,datta2012one,datta2013smooth,datta2014limit}, together with a quantum decoupling result for a measurement channel. Furthermore, by virtue of a proposed reverse Pinsker inequality (Lemma \ref{lem-revPinsker}), a similar one-shot relative-entropy bound for soft covering can be obtained in a straightforward manner. However, this bound can be considerably loose due to its dependence on the minimal eigenvalue of the target state. A detailed comparison between these two one-shot results is provided in Section \ref{sec-main-compare}.

It is noteworthy that, in our proof of the quantum soft covering problem, quantum decoupling \cite{dupuis2014one}\cite{dupuis-thesis}\cite{horodecki2007quantum,abeyesinghe,hayden2008decoupling,berta2011quantum,majenz2017catalytic,wakakuwa2021oneshot,dupuis2023privacy,li2024reliability,cheng2024joint} serves as a powerful tool, since an achievable quantum soft covering protocol approximately realizes a decoupling between the reference system and the target system with respect to a measurement channel. Similar observations have also been made in the quantum channel coding problem, e.g., \cite{hayden2008decoupling}. This indicates that our formulation of soft covering can be considered as a dual to the channel coding framework. On the other hand, the notion of quantum decoupling studied in the literature is again grounded on a trace-distance criterion. Hence, to fully leverage this decoupling tool in our relative-entropy soft covering setting, based on the formulation, we also establish a one-shot quantum decoupling result with relative entropy criterion (Theorem \ref{thm-decouple-1shot}). 

Furthermore, in the classical and CQ covering problem, when the code rate exceeds the corresponding rate threshold (mutual information for classical channels or Holevo information for CQ channels), the covering error decays exponentially fast, and an achievable error exponent (also referred to as the reliability function) has been investigated in the classical \cite{hayashi2006general}\cite{yu2018renyi}\cite{yu2025renyi}\cite{li2025two}\cite{yassaee2019almost,yagli2019exact,parizi2016exact} and CQ \cite{cheng2023error,matsuura2025universal,cheng2025sharp} settings. 
Nevertheless, because soft covering for fully quantum channels has been introduced only recently \cite{atif2024quantum}, the characterization of its error exponent remains open, even in the case of trace-distance covering. In this work, we propose an achievable error exponent for fully quantum soft covering under both trace-distance (Theorem \ref{thm-cover-errexp-l1}) and relative-entropy criteria (Corollary \ref{cor-covering-errexp-D}). The approach again leverages the connection between decoupling and soft covering, as mentioned before: a good covering is achieved when the measurement outcome on the reference system is approximately decoupled from the target state. Hence, our error exponent is obtained by applying recent results on the error exponents of quantum decoupling (e.g., \cite{dupuis2023privacy}\cite{li2024reliability}, in particular \cite{cheng2024joint}) to a measurement channel. Our error exponent also leads to tighter second-order rates (Theorem \ref{thm-cover-2nd-l1} and Theorem \ref{thm-cover-2nd-D}) than those obtained directly from the smoothed entropy approach in \cite{atif2024quantum}. 
Interestingly, when switching to relative-entropy criteria for the soft covering, the corresponding trace-distance exponents remain asymptotically achievable, thanks to the reverse Pinsker inequality. While the reverse Pinsker inequality performs poorly in the one-shot setting, it serves as a powerful tool in the first- and second-order asymptotic regimes.

This paper is organized as follows. Section \ref{sec-preli} introduces useful definitions and lemmas. The formulations and the main results are presented in Section \ref{sec-results}. Section \ref{sec-decouple} proves our conclusions for quantum decoupling and Section \ref{sec-cover} proves those for quantum soft covering.

\section{Preliminaries}\label{sec-preli}
We use the following notations. Let $\mH_A$ denote the Hilbert space of system $A$. Let $I_A$ denote the identity operator on $\mH_A$. Let $\id_A$ denote the identity channel. 
Let $\pi_A := \frac{I_A}{|\mH_A|}$ denote the maximally mixed state on system $A$. Moreover,
let $\mL(\mH_A)$ denote the set of linear operators.  
Let $\mP(\mH_A)$ denote the set of positive operators. We define $\mD(\mH_A) := \{\rho \in \mP(\mH_A): \tr[\rho] = 1\}$ as the set of density operators (quantum states), and $\mD_\leq(\mH_A) := \{\rho \in \mP(\mH_A): \tr[\rho] \leq 1\}$ as the set of subnormalized states. 
For a Hilbert space $\mH_A$, $|\mH_A| := \dim(\mH_A)$ denotes its dimension. For an operator $\rho\in\mL(\mH_A)$, $|\rho|$ denotes its rank and $\supp(\rho)$ denotes its support. 
Completely positive trace-preserving maps are called CPTP maps. In addition, we define $[x]^+ := \max\{x,0\}$, and write $\log:=\log_2$ and $\ln:=\log_e$. The notation $\mathbb{E}_U$ denotes the expectation with respect to the Haar measure over the full unitary group. Throughout this paper, the function $\ell(\cdot)$ is defined by $\ell(\epsilon) := -\epsilon \log\epsilon$.

\subsection{Useful definitions}

For the reader’s convenience, we recall the definitions of a collection of well-known information-theoretic quantities.

\begin{definition}[Relative entropy]\label{def-D}
    Given $\sigma,\rho \in \mD_\leq(\mH)$ such that $\supp(\sigma) \subset \supp(\rho)$, their relative entropy is defined as 
    $$ D(\sigma\|\rho) := \tr\left[ \sigma\left(\log\sigma - \log\rho\right) \right]. $$
\end{definition}

Throughout this paper, we follow Definition \ref{def-D} of relative entropy even for unnormalized states $\sigma$ and $\rho$, i.e., we do not include the factor $\tr[\sigma]$, as is done in many other works, in order to simplify notation. The reason for this is that the relative entropy for unnormalized states appears only as an intermediate quantity in the continuity argument (see Lemma \ref{lem-contin-rela}). All relative-entropy expressions in our main results involve normalized states.

\begin{definition}[Coherent information]\label{def-Ic}
    Given an input state $\rho_A\in\mD(\mH_A)$ and a quantum channel $\mN\colon \mL(\mH_A)\to\mL(\mH_B)$, 
    the coherent information is defined as
    $$ I_c(\rho_A,\mN) 
    = I(R\rangle B)_{\rho_{BR}} 
    := D(\rho_{BR} \| \rho_B \otimes I_R) 
    = H(B) - H(BR), $$
    where $\rho_{BR} := (\mN\otimes\id_R) (\Psi_{AR})$, and $\Psi_{AR}$ is a purification of $\rho_A$.
\end{definition}

\begin{definition}[Quantum information variance]\label{def-vari}
    Given $\sigma\in\mD(\mH)$ and $\rho\in\mP(\mH)$ such that $\supp(\sigma) \subset \supp(\rho) $, their quantum information variance is defined as
    $$ V(\sigma\|\rho) = \tr\left[\sigma(\log\sigma - \log\rho)^2 \right] - D^2(\sigma\|\rho). $$
\end{definition}

\begin{definition}\label{def-V}
    Given an input state $\rho_A\in\mD(\mH_A)$ and a quantum channel $\mN\colon \mL(\mH_A)\to\mL(\mH_B)$,
    we write
    $$ V(\rho_A,\mN) 
    := V(\rho_{BR}\|\rho_B\otimes I_R), $$
    where $V(\cdot\|\cdot)$ is the information variance in Definition \ref{def-vari}, $\rho_{BR} := (\mN\otimes\id_R) (\Psi_{AR})$, and $\Psi_{AR}$ is a purification of $\rho_A$.
\end{definition}

\begin{definition}[Fidelity and Purified Distance]
    Given $\tau,\rho \in \mD_\leq(\mH)$, their fidelity is defined as
    $$ F_*(\tau\|\rho) = \left[ \left\| \sqrt{\tau}\sqrt{\rho} \right\|_1 + \sqrt{\left(1-\tr[\tau]\right) \left(1-\tr[\rho]\right)} \right]^2, $$
    and their purified distance is defined as
    $P(\tau,\rho) = \sqrt{1-F_*(\tau\|\rho)}$.
    Moreover, given $\rho\in \mD_\leq(\mH)$, the $\epsilon$-ball at $\rho$ is defined as $ \mathscr{B}_\leq^\epsilon(\rho) 
        = \{\tau\in\mD_\leq(\mH): P(\tau,\rho) \leq \epsilon \}. $
\end{definition}

\begin{definition}[Quantum R\'enyi divergence] \label{def-Da}
    Let $\alpha \in (0,1)\cup(1,\infty)$. Given $\sigma,\rho \in \mD_\leq(\mH)$ such that $\supp(\sigma) \subset \supp(\rho)$, their Petz–R\'enyi divergence $D_\alpha$ and sandwiched R\'enyi divergence $\tlD_\alpha$ are defined as, respectively,
    $$ D_\alpha (\sigma\|\rho) 
        = \frac{1}{\alpha-1} \log \frac{ \tr\left[ \sigma^\alpha \rho^{1-\alpha} \right] } {\tr[\sigma]}, \quad
    \tlD_\alpha (\sigma\|\rho) 
        = \frac{1}{\alpha-1} \log \frac{\tlQ_\alpha (\sigma\|\rho)}{\tr[\sigma]}, $$
    where
    $ \tlQ_\alpha (\sigma\|\rho) 
    = \tr \left[\left( \rho^{\frac{1-\alpha}{2\alpha}} \sigma \rho^{\frac{1-\alpha}{2\alpha}} \right)^\alpha \right] $.
\end{definition}

\begin{definition}[Max-R\'enyi divergence]\label{def-Dmax}
    Given $\sigma,\rho \in \mD(\mH)$ such that $\supp(\sigma) \subset \supp(\rho)$, the max-R\'enyi divergence is defined as
    $ D_{\max}(\sigma\|\rho)
    := \inf \left\{\lambda\in\mathbb{R}: \sigma \leq 2^\lambda \rho\right\}$. Note that $D_{\max}(\sigma\|\rho) = \tlD_\infty (\sigma\|\rho)$ \cite[Def.~4.8]{tomamichel-book}.
\end{definition}

\begin{definition}[Smooth max-R\'enyi divergence\text{\cite[Eq.~(6.80)]{tomamichel-book}} ] \label{def-Dmax^ep}
    Given $\sigma\in\mD_\leq(\mH), \rho\in\mD(\mH)$ and $\epsilon\in[0,\sqrt{\tr[\sigma]}]$, the $\epsilon$-smooth max-R\'enyi divergence is defined as
    $$ D_{\max}^\epsilon(\sigma\|\rho) = \min_{\hsig\in\mathscr{B}_\leq^\epsilon(\sigma)} D_{\max}(\hsig\|\rho). $$
\end{definition}

\begin{lemma}[AEP of the smooth divergence\text{\cite[Eq.~(6)]{tomamichel2013hierarchy}}]\label{lem-Dmax-AEP}
    $$ D_{\max}^\epsilon(\sigma^\tenn \| \rho^\tenn) 
    = n D(\sigma\|\rho) - \sqrt{n V(\sigma\|\rho)} \ \Phi^{-1}(\epsilon^2) + O \left( \log n \right), $$
    where $ \Phi(x) = \dfrac{1}{\sqrt{2\pi}} \displaystyle{\int_{-\infty}^x} e^{-t^2/2} dt $.
\end{lemma}

\begin{definition}[$\alpha$-conditional entropy\text{\cite[Def.~5.2]{tomamichel-book}}] \label{def-Ha}
    Given $\rho_{AB} \in \mD_\leq(\mH_A \otimes \mH_B)$, we define the following $\alpha$-R\'enyi conditional entropies of $A$ conditioned on $B$:
    \begin{align*}
        \tlH_\alpha^\downarrow(A|B)_{\rho_{AB}} 
            = - \tlD_\alpha(\rho_{AB}\|I_A \otimes \rho_B), \quad
        \tlH_\alpha^\uparrow(A|B)_{\rho_{AB}} 
            = - \inf_{\sigma_B \in \mD(\mH_B)} \tlD_\alpha(\rho_{AB}\|I_A \otimes \sigma_B).
    \end{align*}
\end{definition}

\begin{definition}[Min-entropy\text{\cite[Def.~6.2]{tomamichel-book}}] \label{def-Hmin}
    Given $\rho_{AB} \in \mD_\leq(\mH_A \otimes \mH_B)$, the min-entropy of $A$ conditioned on $B$ of $\rho_{AB}$ is defined as
    $$ H_{\min}(A|B)_{\rho_{AB}} : = - \inf_{\sigma_B \in \mD(\mH_B)} D_{\max}(\rho_{AB}\|I_A \otimes \sigma_B). $$
    Note that $H_{\min}(A|B)_{\rho_{AB}} = \tlH_\infty^\uparrow(A|B)_{\rho_{AB}}$ \cite[Def.~6.2]{tomamichel-book}.
\end{definition}

\begin{definition}[Smooth min-entropy\text{\cite[Def.~6.9]{tomamichel-book}}] \label{def-Hmin^ep}
    Given $\rho_{AB} \in \mD_\leq(\mH_A \otimes \mH_B)$ and $\epsilon>0$, the $\epsilon$-smooth min-entropy of $A$ conditioned on $B$ of $\rho_{AB}$ is defined as
    $$ H_{\min}^\epsilon(A|B)_{\rho_{AB}} = \max_{\hrho_{AB} \in \mathscr{B}_\leq^\epsilon(\rho_{AB})} H_{\min}(A|B)_{\hrho_{AB}}. $$
\end{definition}

\subsection{Useful lemmas}

This subsection presents useful lemmas that will be used to prove our main results.

\begin{lemma}[Quantum Pinsker inequality\text{\cite[Thm.~11.9.1]{wildebook}}] \label{pinsker}
    Given $\sigma\in\mD(\mH)$ and $\rho\in\mD_\leq(\mH)$ such that $\supp(\sigma)\subset\supp(\rho)$, we have
    $$ D(\sigma\|\rho) \geq \frac{1}{2\ln2} \left\|\sigma - \rho \right\|_1^2. $$
\end{lemma}

We present the following continuity bound for quantum relative entropy. It is analogous to the Fannes-Audenaert inequality \cite[Thms.~11.10.1 and 11.10.2]{wildebook}, which provides a continuity bound for the von Neumann entropy. This is an extension to subnormalized random states.

\begin{lemma}[Continuity of quantum relative entropy]\label{lem-contin-rela}
    Given $\rho, \hat{\rho} \in \mD_\leq(\mH)$ such that $\left\|\rho - \hat{\rho}\right\|_1 \leq \epsilon_\rho \leq \frac1e $ and two random operators $\sigma, \hsig \in \mD_\leq(\mH)$ such that
    $ \mathbb{E} \sigma = \rho$,  
    $\mathbb{E} \hsig = \hat{\rho}$, and 
    $\mathbb{E}\left\|\sigma - \hsig\right\|_1 \leq \epsilon_\sigma \leq \frac1e$,
    we have
    $$ \big| \mathbb{E}D(\sigma\|\rho) - \mathbb{E}D(\hsig\|\hat{\rho}) \big| \leq (\epsilon_\sigma + \epsilon_\rho) \log|\mH| + \ell(\epsilon_\sigma) + \ell(\epsilon_\rho). $$
\end{lemma}


Lemma \ref{lem-contin-rela} is proved in Appendix \ref{app-proof-lem-contin}. In many quantum information problems based on the trace-distance criterion, the following quadratic bound arises \cite[Lem.~5.1.3]{renner2008security}\cite[Lem.~3.6]{dupuis2014one}:
    $$ \left\|\sigma-\rho\right\|_1 \leq \sqrt{\tr[\xi] \ \tr\left[(\sigma-\rho) \xi^{-\frac12} (\sigma-\rho) \xi^{-\frac12}\right]}, \quad \forall \ \xi\in\mD(\mH). $$
In this work, which focuses on the relative entropy criterion, we establish a similar quadratic bound for the relative entropy in Lemma \ref{lemma-D<Q2}.

\begin{lemma}[Upper bound of relative entropy]\label{lemma-D<Q2}
    Given $\sigma, \rho \in \mD_\leq(\mH)$ such that $\supp(\sigma)\subset\supp(\rho)$, we have
    $$ \frac{1}{\log e} D(\sigma\|\rho) \leq \tr\left[\sigma \rho^{-\frac12} \sigma \rho^{-\frac12} \right] - \tr[\sigma]. $$
\end{lemma}

\begin{proof}
This follows from directly $\tlD_1(\sigma\|\rho) \leq \tlD_2(\sigma\|\rho)$ \cite[Cor.~4.13]{tomamichel-book} and $\ln(x)\leq x-1$. 
\end{proof}

Moreover, the duality relations of R\'enyi quantities are also well-studied in the literature; see, e.g., \cite[Sec.~5.3]{tomamichel-book}. Here we introduce a duality bound between $\tlQ_2$ and $D_{\max}$ which applies to unnormalized states.

\begin{lemma}[Duality bound for unnormalized states]\label{lem-duality}
    Given $\rho_{AB} \in \mP(\mH_A \otimes \mH_B)$ and $\tau_A \in \mP(\mH_A)$ such that $\supp(\rho_A)\subset\supp(\tau_A)$, we have
    $$ \log \tlQ_2(\rho_{AB} \| \tau_A \otimes \rho_B) \leq \inf_{\xi_B \in \mD(\mH_B)} D_{\max}(\rho_{AB}\|\tau_A \otimes \xi_B) $$
\end{lemma}

\begin{remark}
    Taking $\tau_A = I_A$, we get that $-\tilde{H}_2^\downarrow(A|B) \leq -\tilde{H}_\infty^\uparrow(A|B)$, which is \cite[Eq.~(48)]{tomamichel2014relating}.
\end{remark}

Lemma \ref{lem-duality} is proved in Appendix \ref{app-proof-duality}. Next, we propose a reverse Pinsker inequality that bounds relative entropy by the trace distance. 

\begin{lemma}[A reverse Pinsker inequality]\label{lem-revPinsker}
    Given $\sigma,\rho \in \mD(\mH)$ such that $\supp(\sigma) \subset \supp(\rho) $ and $\left\|\sigma-\rho\right\|_1 \leq \epsilon \leq \frac{1}{e}$, we have 
    $$ D(\sigma\|\rho) \leq \epsilon \log \frac{|\mH|}{\eigrho} + \ell(\epsilon), $$
    where $\eigrho$ is the smallest eigenvalue of $\rho$.
\end{lemma}

Lemma \ref{lem-revPinsker} is proved in Appendix \ref{app-proof-revPinsker}. A natural corollary from Lemma \ref{lem-revPinsker} is that in the $n$-shot setting, if $\rho$ is a product state, then $\log \frac{|\mH|}{\eigrho}$ contributes only linearly in $n$. Hence, if $\epsilon$ is exponentially small, then so is the relative entropy. We formalize this observation in the following Lemma \ref{lem-errexp-rela}, which states that the error exponent of trace distance implies an identical exponent of the relative entropy. 

\begin{lemma}[Error exponent of trace distance implies that of relative entropy]\label{lem-errexp-rela}
    Given $\rho\in\mD(\mH)$ and a sequence $\sigma_n\in \mD(\mH^\tenn)$ such that $\supp(\sigma_n) \subset \supp(\rho^\tenn) $ and $\left\|\sigma_n - \rho^\tenn\right\|_1 \leq 2^{-nE}$ for some constant $E>0$, then we have 
    $ \limsup_{n\to\infty} -\frac1n \log D\left(\sigma_n\|\rho^\tenn\right) \geq E. $
\end{lemma}

\begin{proof}
    Let $\epsilon = 2^{-nE}$. Then $\ell(\epsilon) = - \epsilon \log \epsilon = nE \cdot 2^{-nE}$. Also, note that $\eig(\rho^\tenn)_{\min} = (\eigrho)^n$. Then our conclusion is immediately implied by Lemma \ref{lem-revPinsker}.
\end{proof}

It is noteworthy that Lemma \ref{lem-errexp-rela} is an asymptotic result. In the one-shot scenario, this reverse Pinsker inequality can still yield a loose error bound (see Section \ref{sec-main-compare} for further discussion), which motivates our derivation of a significantly more efficient one-shot result under the relative entropy criterion.

\section{Main Results}\label{sec-results}

In this section, we present the main results of this work. Section \ref{sec-main-decouple} summarizes our one-shot achievability for quantum decoupling with relative entropy criterion. Section \ref{sec-main-cover} formulates quantum soft covering and presents a series of results, including one-shot bounds, asymptotic bounds, error exponents, and second-order rates. In Section \ref{sec-main-compare}, we compare our one-shot soft covering bound with the bound obtained via the proposed reverse Pinsker inequality, and show that the former is tighter and superior.

\subsection{Quantum Decoupling}\label{sec-main-decouple}

Following \cite{dupuis2014one}, we formulate the quantum decoupling problem as follows. 
Consider a bipartite state $\rho_{AE} \in \mD(\mH_A \otimes \mH_E)$ and a quantum channel (CPTP map) $\mT \colon \mL(\mH_A)\to\mL(\mH_B)$. Our objective is to decouple the output of the channel from system $E$ by unitarily modifying $n$-copies of system $A$ before feeding it to the channel. 

Let $\mT^\tenn \circ U_{A^n}$ act on system $A$, where $U_{A^n}$ is a unitary operator acting on $\mH_A^\tenn$. At the output of $\mT^\tenn$, we obtain a state
    \begin{equation}\label{sig_BE}
        \sigma_{B^nE^n} = (\mT^\tenn \otimes \id_{E^n}) 
            \left( (U_{A^n} \otimes I_{E^n})
            \rho_{AE}^\tenn
            (U_{A^n}^\dagger \otimes I_{E^n}) \right).
    \end{equation}
Let $|\mH_A| = d$. Define the Choi map of $\mT$ as 
\begin{equation}\label{tau_AB}
    \tau_{AB} := \frac{1}{d} (\id_A \otimes \mT) \ket{\Gamma}_{AA'} \bra{\Gamma}.
\end{equation}
Here 
$ \ket{\Gamma}_{AA'} = \sum_i \ket{i}_A\ket{i}_{A'}$ is the unnormalized maximally entangled state, where $\{\ket{i}\}$ denoting a basis of $\mH_A$ and $A'$ is a copy of system $A$. One can check that $\tau_B = \mT(\pi_A)$. Also, define
\begin{equation}\label{rho_BE}
    \rho_{BE} := \tau_B \otimes \rho_E = \mT(\pi_A) \otimes \rho_E.
\end{equation}
The above construction is illustrated by Figure \ref{fig-decouple}. The quantum decoupling problem asks how close $\sigma_{B^nE^n}$ can be made to $\rho_{BE}^\tenn$.

\begin{figure}[h]
  \centering
  \includegraphics[width=0.4\textwidth]{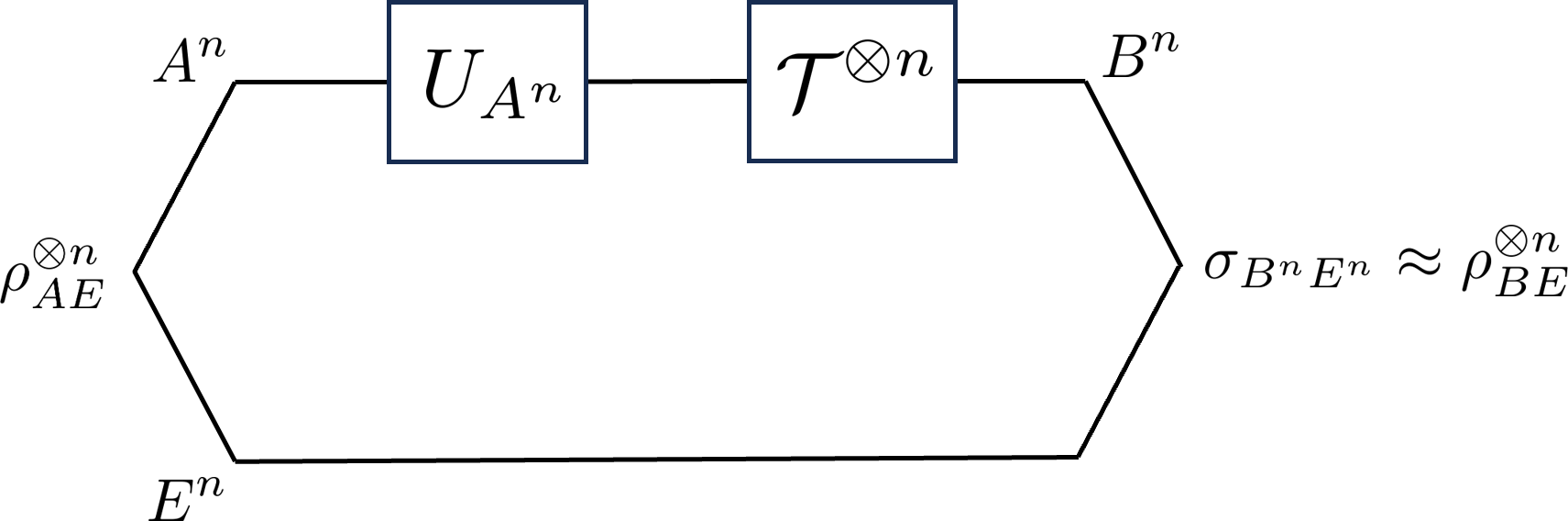}
  \caption{Quantum Decoupling}
  \label{fig-decouple}
\end{figure}

\begin{definition}[Quantum decoupling]
    A quantum decoupling setup is characterized by a pair $(\rho_{AE}, \mT)$ where $\rho_{AE} \in \mD(\mH_A\otimes\mH_E)$ is a state and $\mT: \mL(\mH_A)\to\mL(\mH_B)$ is a CPTP map. 
\end{definition}

\begin{definition}[Quantum decoupler]
    Given a pair $(\rho_{AE}, \mT)$, an $(n,\delta)$ quantum decoupler is a unitary operator $U_{A^n}$ acting on $\mH_A^\tenn$ such that $d( \sigma_{B^nE^n}, \rho_{BE}^\tenn) \leq \delta$, where $\sigma_{B^nE^n}$ and $\rho_{BE}$ are defined in \eqref{sig_BE} and \eqref{rho_BE}, respectively. 
    Here $d(\cdot,\cdot)$ is some given criterion. 
    If $d(\sigma,\rho) = \left\|\sigma-\rho\right\|_1$ is the trace distance, then the corresponding quantum decoupler is called an $(n,\delta)_1$ (trace-distance) quantum decoupler.
    If $d(\sigma,\rho) = D(\sigma\|\rho)$ is the relative entropy, then the corresponding quantum decoupler is called an $(n,\delta)_D$ (relative-entropy) quantum decoupler.
\end{definition}

A one-shot achievablity for $(1,\delta)_1$ trace-distance quantum decouplers is investigated in \cite[Thms.~3.1]{dupuis2014one}. In this work, we focus on the relative-entropy quantum decoupler. We state the following theorem on the one-shot achievability, which is proved in Section \ref{sec-decouple}. 

\begin{theorem}[one-shot quantum decoupling]\label{thm-decouple-1shot}
    Given a pair $(\rho_{AE},\mT)$, for any $\epsilon\in(0,\frac1{8e})$, there exists a $(1,\delta)_D$ quantum decoupler such that 
    $$ \delta \leq (\log e) \exp_2 
        \left\{ -H_{\min}^\epsilon(A|B)_{\tau_{AB}} -H_{\min}^\epsilon(A|E)_{\rho_{AE}}\right\} 
        + 12\epsilon \log|\mH_B| |\mH_E| + \ell(8\epsilon) + \ell(4\epsilon). $$
    where $H_{\min}^\epsilon$ is the smooth min-entropy in Definition \ref{def-Hmin^ep}.
\end{theorem}

\subsection{Quantum Soft Covering}\label{sec-main-cover}


This subsection presents our formulations and results of the quantum soft covering. In the general problem setting, we consider sending a rank-constrained (the rank can be understood as the code rate) encoded state through a quantum channel $\mN(\cdot)\colon \mL(\mH_A)\to\mL(\mH_B)$, by using it $n$ times, in order to approximate a given product output state $\rho_B^\tenn$. We begin by formulating the quantum soft covering problem as follows.

\begin{definition}[Quantum soft covering]
    A quantum covering setup is characterized by a pair $(\rho_B, \mN)$ where $\rho_B \in \mD(\mH_B)$ is a state and $\mN: \mL(\mH_A)\to\mL(\mH_B)$ is a CPTP map. Define 
    $ \mS(\rho_B,\mN) := 
    \{\rho_A \in \mD(\mH_A): \mN(\rho_A)=\rho_B\} $
    to be the set of all input states that yield that target output $\rho_B$ under channel $\mN$.
    For an input state $\rho_A\in\mS(\rho_B,\mN)$, we write $\rho_{BR} := (\mN\otimes\id_R) (\Psi_{AR})$, where $\Psi_{AR}$ is a purification of $\rho_A$.
\end{definition}

\begin{definition}[Quantum soft covering code]
    Given a pair $(\rho_B, \mN)$, an $(n,\Theta,\delta)$ quantum soft covering code is a state $\sigma_{A^n} \in \mD(\mH_A^\tenn)$ such that $|\sigma_{A^n}|=\Theta$ and $d \left( \mN(\sigma_{A^n} \right), \rho_B^\tenn) \leq \delta$. Here $d(\cdot,\cdot)$ is some given criterion. 
    If $d(\sigma,\rho) = \left\|\sigma-\rho\right\|_1$ is the trace distance, then the corresponding quantum covering code is called an $(n,\Theta,\delta)_1$ (trace-distance) quantum covering code. 
    If $d(\sigma,\rho) = D(\sigma\|\rho)$ is the relative entropy, then the corresponding quantum covering code is called an $(n,\Theta,\delta)_D$ (relative-entropy) quantum covering code. 
\end{definition}

The one-shot and asymptotic characterizations of trace-distance quantum covering are discussed in \cite[Thms.~3.3 and 3.5]{atif2024quantum}. Hence, in this work, we focus on these two characterizations of quantum soft covering with relative entropy criterion. We first present the following Theorem \ref{thm-cover-1shot} on the one-shot result. A proof is available in Section \ref{sec-cover-rela}.

\begin{theorem}[One-shot quantum covering]\label{thm-cover-1shot}
    Given a pair $(\rho_B, \mN)$, for any $\rho_A \in \mS(\rho_B,\mN)$,  
    for all $\epsilon\in(0,\frac{1}{2e})$ and $\eta\in(0,1)$, there exists an $(1,\Theta,\delta)_D$ quantum covering code such that 
    \begin{align*}
        \log\Theta &\leq 
            \left[ -H_{\min}^\epsilon(R|B)_{\rho_{BR}} - \log\eta \right]^+, \\
        \delta &\leq 
            2(\log e)\eta + 4\epsilon \log|\mH_A||\mH_B| + 2\ell(2\epsilon),
    \end{align*}
    where $H_{\min}^\epsilon(R|B)$ is the smooth min-entropy in Definition \ref{def-Hmin^ep}. Moreover, for all $\delta \in (0,1)$, every $(1,\Theta,\delta)_D$ quantum covering code satisfies 
    $$ \log \Theta \geq \inf_{\rho_A' \in \mS_{\delta}(\rho_B,\mN)} \left[ -H_{\min}(R|B)_{\rho'_{BR}} \right]^+, $$
    where $\mS_{\delta}(\rho_B,\mN) := \left\{ \rho_A' \in \mD(\mH_A) : D(\mN(\rho'_A) \| \rho_B) \leq \delta \right\}$, and $\rho'_{BR} := (\mN \otimes \id_R) \Psi'_{AR}$ with $\Psi'_{AR}$ being a purification of $\rho_A'$. 
\end{theorem}

Extending this one-shot result to the asymptotic scenario, we obtain Theorem \ref{thm-cover-asym} below, which characterizes the lowest achievable rate. This rate coincides with that of the trace-distance case \cite[Thm.~3.5]{atif2024quantum} and is determined by the minimal coherent information required to produce the target output. Unlike quantum channel coding, the limit of the asymptotically achievable rate in the quantum soft covering ultimately reduces to a single-letter expression.

\begin{definition}[Achievability of quantum covering]
    Given a pair $(\rho_B, \mN)$, a rate $r$ is said to be achievable for quantum covering if for all $\delta>0$ and all sufficiently large $n$, there exists an $(n,\Theta,\delta)_D$ quantum covering code such that $\frac{1}{n}\log\Theta \leq r + \delta$.
\end{definition}

\begin{theorem}[Asymptotic quantum covering]\label{thm-cover-asym}
    Given a pair $(\rho_B, \mN)$, we have
    $$ \inf\{r: r\text{ is achievable}\} 
    = \inf_{\rho_A \in \mS(\rho_B,\mN)} 
        \big[ I_c(\rho_A,\mN) \big]^+, $$
    where $I_c(\rho_A,\mN)$ is the coherent information, see Definition \ref{def-Ic}.
\end{theorem}

An outline of the proof of Theorem \ref{thm-cover-asym} is provided in Section \ref{sec-cover-asym}. Furthermore, we establish achievable error exponents for quantum soft covering with trace-distance and relative-entropy criteria. 

\begin{theorem}[One-shot error exponent of quantum covering with trace distance criterion] \label{thm-cover-errexp-l1}
    Given a pair $(\rho_B, \mN)$, for any $\rho_A \in \mS(\rho_B,\mN)$,
    there exists an $(1,2^r,\delta)_1$ quantum covering code such that 
    $$ -\log\frac{\delta}{4} \geq E(r) 
    := \sup_{\alpha\in[1,2]} \frac{\alpha-1}{\alpha} \left( r
    + \tilde{H}_\alpha^\uparrow(R|B)_{\rho_{BR}}\right). $$
    where $\tilde{H}_\alpha^\uparrow(R|B)_{\rho_{BR}}$ is the $\alpha$-conditional entropy, see Definition \ref{def-Ha}.
\end{theorem}

A proof is provided in Section \ref{sec-cover-errexp}. It is noteworthy that $\rho_A^\tenn\in\mS(\rho_B^\tenn,\mN^\tenn)$ if $\rho_A\in\mS(\rho_B,\mN)$, resulting in an additive $\tilde{H}_\alpha^\uparrow$, i.e., $\tilde{H}_\alpha^\uparrow(R|B)_{\rho_{BR}^\tenn} = n \tilde{H}_\alpha^\uparrow(R|B)_{\rho_{BR}}$ \cite[Cor.~5.9]{tomamichel-book}. Hence, Theorem \ref{thm-cover-errexp-l1} readily generalizes to the $n$-shot setting. Combining this with Lemma \ref{lem-errexp-rela}, one immediately obtains the following error exponent of quantum soft covering with relative entropy criterion. 

\begin{corollary}[Error exponent of quantum covering with relative entropy]\label{cor-covering-errexp-D}
    Given a pair $(\rho_B, \mN)$, for any $\rho_A \in \mS(\rho_B,\mN)$, 
    there exists an $(n,2^{nr},\delta_n)_D$ quantum covering code such that 
    $$-\frac1n \limsup_{n\to\infty} \log \delta_n \geq E(r), $$
    where $E(r)$ is defined in Theorem \ref{thm-cover-errexp-l1}.
\end{corollary}

Two examples of the proposed soft covering error exponent $E(r)$ are presented in the following Figure \ref{fig-covering-errexp}. The first example uses a depolarizing channel, and the second uses an erasure channel. In both examples, the input Hilbert space is two-dimensional (i.e., a qubit space), and the set $\mS(\rho_B,\mN)$ is a singleton $\rho_A = \pi_A$. 

\begin{figure}[!htbp]
    \centering
    \subfloat[Depolarizing channel]{%
        \includegraphics[width=0.49\linewidth]{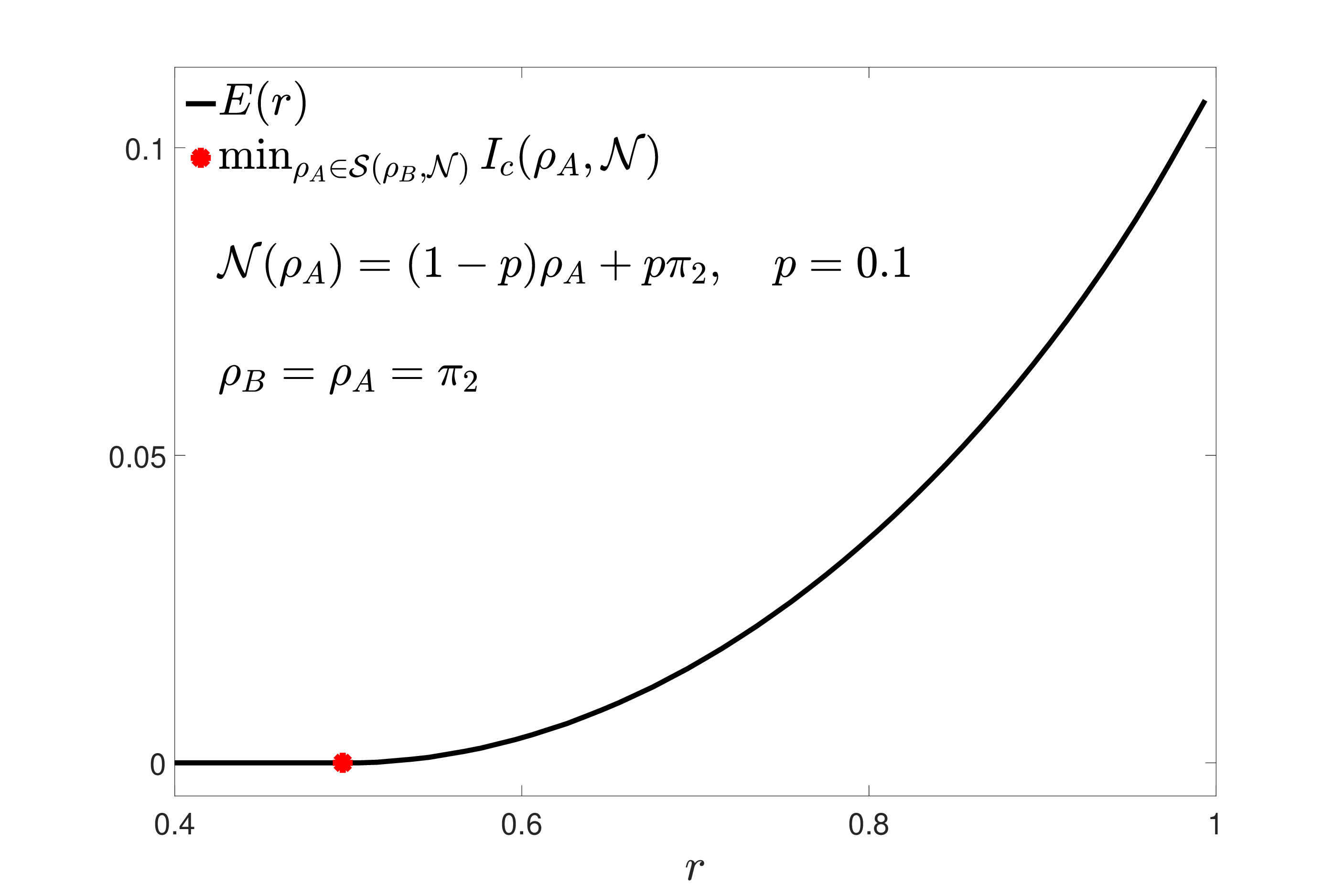}%
        \label{fig-covering-errexp1}
    }\hfill
    \subfloat[Erasure channel]{%
        \includegraphics[width=0.49\linewidth]{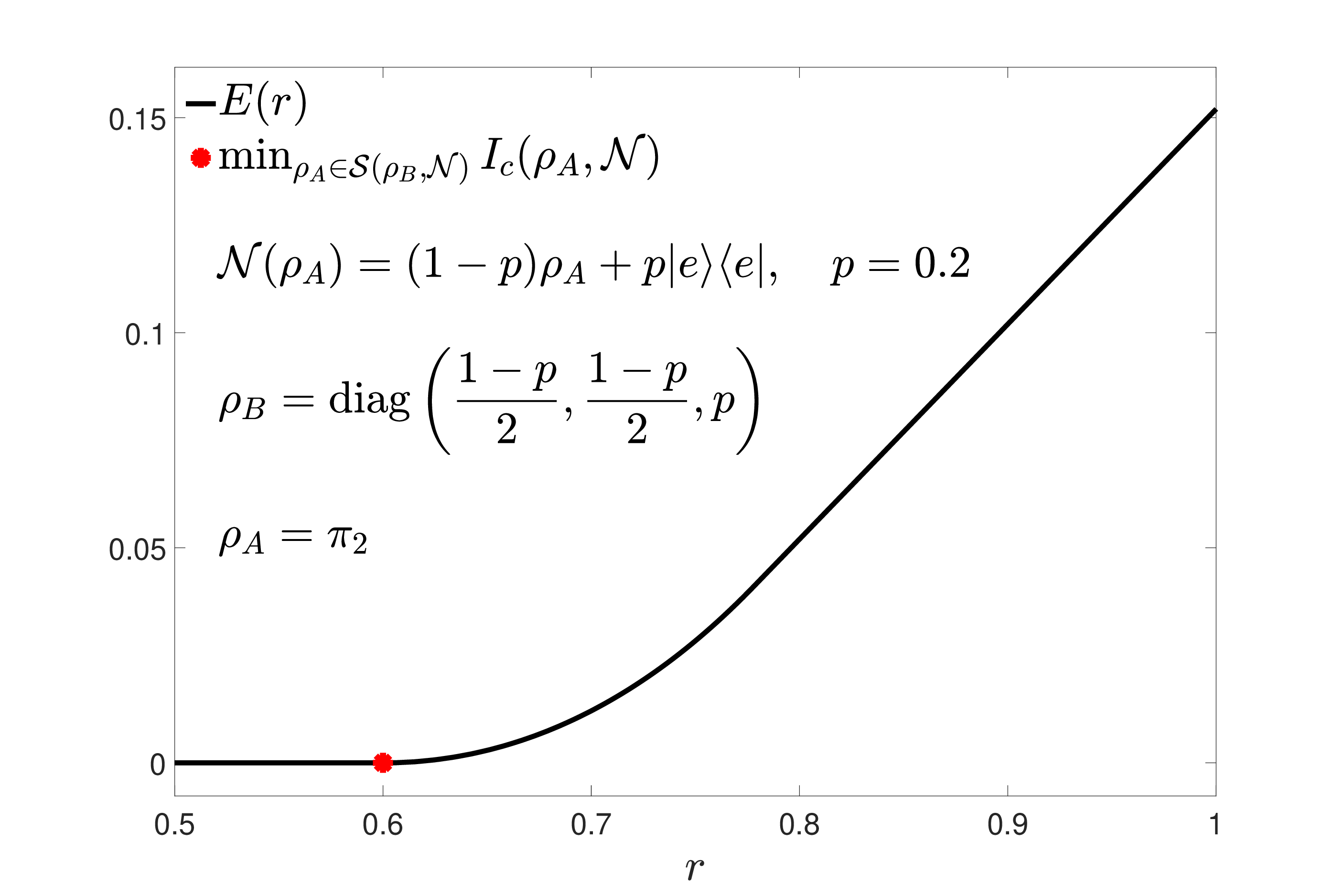}%
        \label{fig-covering-errexp2}
    }
    \caption{Examples of the quantum soft covering error exponent $E(r)$.}
    \label{fig-covering-errexp}
\end{figure}

From the error exponent, we derive second-order achievability bounds for both covering criteria, aimed at providing a more precise estimate of the convergence behavior of the code rate near $I_c(\rho_A,\mN)$ for a given covering error. These bounds are provided in the following Theorem \ref{thm-cover-2nd-l1} and Theorem \ref{thm-cover-2nd-D}.

\begin{theorem}[Second-order achievability with trace distance criterion] \label{thm-cover-2nd-l1}
    Given a pair $(\rho_B, \mN)$, for any $\rho_A\in\mS(\rho_B,\mN)$, there exists an $(n,\Theta,\delta)_1$ quantum covering code with $\delta\in(0,1)$ whose rate scales as
    \begin{equation}\label{2nd-l1}
        \frac{1}{n}\log\Theta 
        = \left[ I_c(\rho_A,\mN) + 
        \sqrt{\left(2V(\rho_A,\mN)\ln\frac{4}{\delta}\right) \frac{1}{n}} 
        + o\left(\frac{1}{\sqrt{n}}\right) \right]^+,
    \end{equation}
    where $I_c(\rho_A,\mN)$ is defined in Definition \ref{def-Ic} and $V(\rho_A,\mN)$ is defined in Definition \ref{def-V}.
\end{theorem}

\begin{theorem}[Second-order achievability with relative entropy criterion] \label{thm-cover-2nd-D}
    Given a pair $(\rho_B, \mN)$, for any $\rho_A\in\mS(\rho_B,\mN)$, there exists an $(n,\Theta,\delta)_D$ quantum covering code with $\delta\in(0,1)$ whose rate scales as
    \begin{equation}\label{2nd-D}
        \frac{1}{n}\log\Theta 
        = \left[ I_c(\rho_A,\mN) + \sqrt{2V(\rho_A,\mN) \ \frac{\ln (cn)}{n}} 
            + o\left(\frac{1}{\sqrt{n}}\right) \right]^+,
    \end{equation}
     where $I_c(\rho_A,\mN)$ is defined in Definition \ref{def-Ic}, $V(\rho_A,\mN)$ is defined in Definition \ref{def-V}, $c = \frac{8\log|\mH_B|}{\eigrho\delta}$ is a constant dependent on $\delta$, and $\eigrho$ is the smallest eigenvalue of $\rho$. 
\end{theorem}

Both Theorem \ref{thm-cover-2nd-l1} and Theorem \ref{thm-cover-2nd-D} are proved in Section \ref{sec-cover-2nd}. In particular, under the setting of Theorem \ref{thm-cover-2nd-l1} with trace distance covering, a similar second-order achievability is proposed by \cite[Thm.~3.7]{atif2024quantum}:
\begin{equation}\label{2nd-atif}
    \frac{1}{n}\log\Theta 
        = \left[ I_c(\rho_A,\mN) 
        + \left|\Phi^{-1}\left(\frac{\delta^2}{100}\right)\right|
        \sqrt{V(\rho_A,\mN) \ \frac{1}{n}} 
        + o\left(\frac{1}{\sqrt{n}}\right) \right]^+.
\end{equation}
We argue that the proposed second-order rate in \eqref{2nd-l1} is tighter than \eqref{2nd-atif}. Using \cite[Prop.~2.1.2]{vershynin-book}
\begin{equation}\label{Phi<ln}
    \left|\Phi^{-1}(\epsilon)\right| \lesssim \sqrt{-2\ln\epsilon},
\end{equation}
the coefficient for $\frac{1}{\sqrt{n}}$ in 
\eqref{2nd-atif} equivalently scales as 
    $$ \sqrt{V(\rho_A,\mN)} \left|\Phi^{-1}\left(\frac{\delta^2}{100}\right)\right|
    \lesssim \sqrt{ 4V(\rho_A,\mN) \ln \frac{10}{\delta} }, $$
which is bigger than the corresponding coefficient $\sqrt{ 2V(\rho_A,\mN) \ln \frac{4}{\delta} }$ in \eqref{2nd-l1}. Ergo, to achieve the same error, the proposed Theorem \ref{thm-cover-2nd-l1} implies the existence of a code with a lower rate. Similarly, Theorem \ref{thm-cover-2nd-D} also yields a tighter second-order rate than that obtained by directly applying the one-shot achievability result in Theorem \ref{thm-cover-1shot}. A detailed discussion is provided in Section \ref{sec-cover-2nd}.

\subsection{Comparison between our one-shot result and the reverse Pinsker result} \label{sec-main-compare}

Letting $\Theta = 2^r$, we rewrite the achievability part of our one-shot Theorem \ref{thm-cover-1shot} as
\begin{equation}\label{delta_de}
    D(\sigma_B\|\rho_B) \leq \delta_\de 
        := (\log e) 2^{-(r + H_{\min}^\epsilon(R|B))} + 4\epsilon \log|\mH_A||\mH_B| + 2\ell(2\epsilon),
\end{equation}
for some $\sigma_B = \mN(\sigma_A)$ with $|\sigma_A|\leq\Theta$. The subscript `$\de$' indicates that this bound is obtained by a decoupling method, as implied by the proof of Theorem \ref{thm-cover-1shot} in Section \ref{sec-cover-rela}. On the other hand, a trace-distance bound is proposed in \cite{atif2024quantum}: $\left\|\sigma_B -\rho_B\right\|_1 \leq 2^{-\frac12(R + H_{\min}^\epsilon(R|B))} + 4\epsilon$. Consequently, one can acquire a trivial bound for $D(\sigma_B\|\rho_B)$ through Lemma \ref{lem-revPinsker}:
\begin{equation}\label{delta_rp}
    D(\sigma_B\|\rho_B) \leq \delta_\rp
        := (2^{-\frac12(r + H_{\min}^\epsilon(R|B))} + 4\epsilon) \log \frac{|\mH_B|}{\eig(\rho_B)_{\min}} + \ell(2^{-\frac12(R + H_{\min}^\epsilon(R|B))} + 4\epsilon).
\end{equation}
Likewise, the subscript `$\rp$' indicates that this bound is obtained by a reverse Pinsker inequality. 

We argue that, compared with \eqref{delta_rp}, the proposed bound \eqref{delta_de} enjoys two advantages. First, when both bounds are extended to the finite-shot regime, the dominant term for $\delta_\rp$ is $2^{-\frac12(r + H_{\min}^\epsilon(R|B))}$, while that for $\delta_\de$ is $2^{-(r + H_{\min}^\epsilon(R|B))}$. The latter is significantly smaller due to the absence of the square root. Second, the smallest eigenvalue $\eig(\rho_B)_{\min}$ may take very small values, which can also lead to a large $\delta_\rp$. The following Figure \ref{fig-compare} compares $\delta_\de$ and $\delta_\rp$ for the depolarizing channel shown in Figure \ref{fig-covering-errexp}\subref{fig-covering-errexp1}. The target state $\rho_B$ is indicated in the plot, and its eigenvalues are $0.99$ and $0.01$.

\begin{figure}[h]
  \centering
  \includegraphics[width=0.5\textwidth]{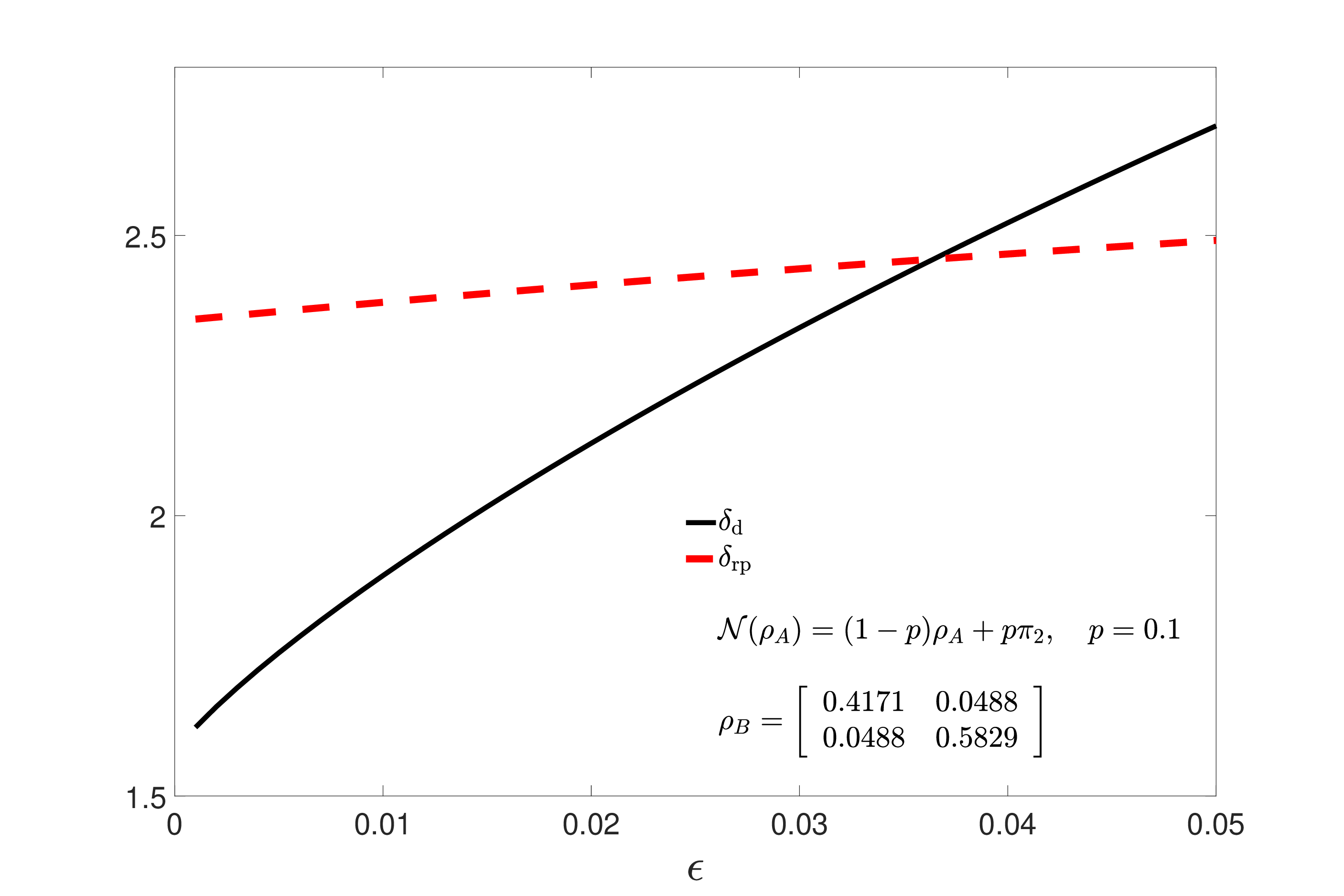}
  \caption{Comparison between $\delta_\de$ and and $\delta_\rp$.}
  \label{fig-compare}
\end{figure}

\section{Quantum Decoupling}\label{sec-decouple}

This section provides a proof of Theorem \ref{thm-decouple-1shot}.

\begin{proof}[Proof of Theorem \ref{thm-decouple-1shot}]
We apply random coding. The proof can be divided into three steps.
The first step is to obtain a $\tlQ_2$ upper bound. Inserting $\sigma_{BE}, \rho_{BE}$ from \eqref{sig_BE} and \eqref{rho_BE} (here $\sigma_{BE}:= \sigma_{B^1E^1}$ is the one-shot output) in Lemma \ref{lemma-D<Q2} yields that
\begin{equation}\label{de:ED<Q2}
    \frac{1}{\log e} \mathbb{E}_U D(\sigma_{BE}\|\rho_{BE}) 
        \leq \Eutr\left[\left(\rho_{BE}^{-\frac14} \sigma_{BE} \rho_{BE}^{-\frac14}\right)^2 \right] - \Eutr[\sigma_{BE}].
\end{equation}
Following the notations in \cite[Sec.~3.3]{dupuis2014one}, we define $\tlmT := \tau_B^{-\frac14} \mT \tau_B^{-\frac14}$ and unnormalized states 
\begin{equation}\label{tlrho-tltau}
    \tlrho_{AE} := (I_A \otimes \rho_E^{-\frac14}) \rho_{AE} (I_A \otimes \rho_E^{-\frac14}), \quad 
    \tltau_{AB} := (I_A \otimes \tau_B^{-\frac14}) \tau_{AB} (I_A \otimes \tau_B^{-\frac14}).
\end{equation}
Tracing out system $A$ gives $\tltau_B = \sqrt{\tau_B}$ and $\tlrho_E = \sqrt{\rho_E}$. Hence, one can further define $\tlrho_{BE} := \tltau_B\otimes\tlrho_E = \sqrt{\rho_{BE}}$ and
\begin{align*}
    \tlsig_{BE} 
    := \ & \rho_{BE}^{-\frac14} \sigma_{BE} \rho_{BE}^{-\frac14} 
    = (\tau_B^{-\frac14} \otimes \rho_E^{-\frac14}) (\mT \otimes \id_E) 
        \left( (U_A \otimes I_E)\rho_{AE} (U_A^\dagger \otimes I_E) \right) 
        (\tau_B^{-\frac14} \otimes \rho_E^{-\frac14}) \\
    = \ & (\tlmT \otimes I_E) 
        \left( (U_A \otimes I_E) \tlrho_{AE} (U_A^\dagger \otimes I_E) \right),
\end{align*}
It is shown in \cite[Sec.~3.3]{dupuis2014one} that $\mathbb{E}_U \sigma_{BE} = \rho_{BE} = \tlrho_{BE}^2$. Then \eqref{de:ED<Q2} reduces to
\begin{align}\label{decouple-1shot-Q2}
    \frac{1}{\log e} \mathbb{E}_U D(\sigma_{BE}\|\rho_{BE})
    &\leq \Eutr\big[\tlsig_{BE}^2\big] - \tr\big[\tlrho_{BE}^2\big] 
    \overset{a}{\leq}
        \tr\big[\tltau_{AB}^2\big]  \tr\big[\tlrho_{AE}^2\big]  \notag\\
    &\overset{b}{=}
        \tlQ_2(\tau_{AB}\| I_A \otimes \tau_B) \ \tlQ_2(\rho_{AE}\| I_A \otimes \rho_E).
\end{align}
Here $(a)$ follows also from \cite[Sec.~3.3]{dupuis2014one}. In $(b)$, we identified the $\tlQ_2$ terms from \eqref{tlrho-tltau}. 

In the second step, we perform smoothing and extend \eqref{decouple-1shot-Q2} to an $H_{\min}^\epsilon$ upper bound. There exist $\htau_{AB} \in \mathscr{B}_\leq^\epsilon(\tau_{AB})$ and $\hrho_{AE} \in \mathscr{B}_\leq^\epsilon(\rho_{AE})$ such that
\begin{equation}\label{HminABandAE}
    H_{\min}(A|B)_{\htau_{AB}} =        
        H_{\min}^\epsilon(A|B)_{\tau_{AB}}, \quad 
    H_{\min}(A|E)_{\hrho_{AE}} = 
        H_{\min}^\epsilon(A|E)_{\rho_{AE}}
\end{equation}
According to \cite[Sec.~3.4]{dupuis2014one}, we can correspondingly construct subnormalized states 
\begin{align*}
    \hsig_{BE} := 
        (\hmT \otimes \id_E)(U_A \otimes I_E) \hrho_{AE}(U_A^\dagger \otimes I_E), \quad 
    \hrho_{BE} := \htau_B\otimes\hrho_E
\end{align*}
with a unique superoperator $\hmT$ such that $\htau_{AB} = \frac{1}{d} (\id_A\otimes \hmT) \ket{\Gamma}_{AA} \bra{\Gamma} $; consequently, 
\begin{equation}\label{de:sig-rho}
    \mathbb{E}_U \left\|\sigma_{BE} - \hsig_{BE}\right\|_1 \leq 8\epsilon, \quad \left\|\rho_{BE} - \hrho_{BE}\right\|_1 \leq 4\epsilon.
\end{equation}
For this new pair $(\hsig_{BE},\hrho_{BE})$, \eqref{decouple-1shot-Q2} gives
\begin{align}\label{de:ED<Dmax^ep}
    \frac{1}{\log e} \mathbb{E}_U D(\hsig_{BE}\|\hrho_{BE}) 
    &\leq \tlQ_2(\htau_{AB}\| I_A \otimes \htau_B) \ \tlQ_2(\hrho_{AE}\| I_A \otimes \hrho_E) \notag\\
    &\overset{a}{\leq} \exp_2 \left\{ - H_{\min}(A|B)_{\htau_{AB}} - H_{\min}(A|E)_{\hrho_{AE}} \right\} \notag\\
    &\overset{b}{=} \exp_2 \left\{ - H_{\min}^\epsilon(A|B)_{\tau_{AB}} - H_{\min}^\epsilon(A|E)_{\rho_{AE}} \right\},
\end{align}
where $(a)$ follows from Lemma \ref{lem-duality} and Definition \ref{def-Hmin}, and $(b)$ follows from \eqref{HminABandAE}.

Finally, we apply the continuity bound. Inserting \eqref{de:sig-rho} and \eqref{de:ED<Dmax^ep} in Lemma \ref{lem-contin-rela} gives our result.
\end{proof}

\section{Quantum Soft Covering}\label{sec-cover}

In this section, we prove our results in quantum soft covering. A key technique used in our proof is random coding. Ergo, we begin with formulating the construction of a random quantum covering code, which is illustrated in Figure \ref{fig-cover-decouple}.                                                                      \begin{figure}[h]
  \centering
  \includegraphics[width=0.5\textwidth]{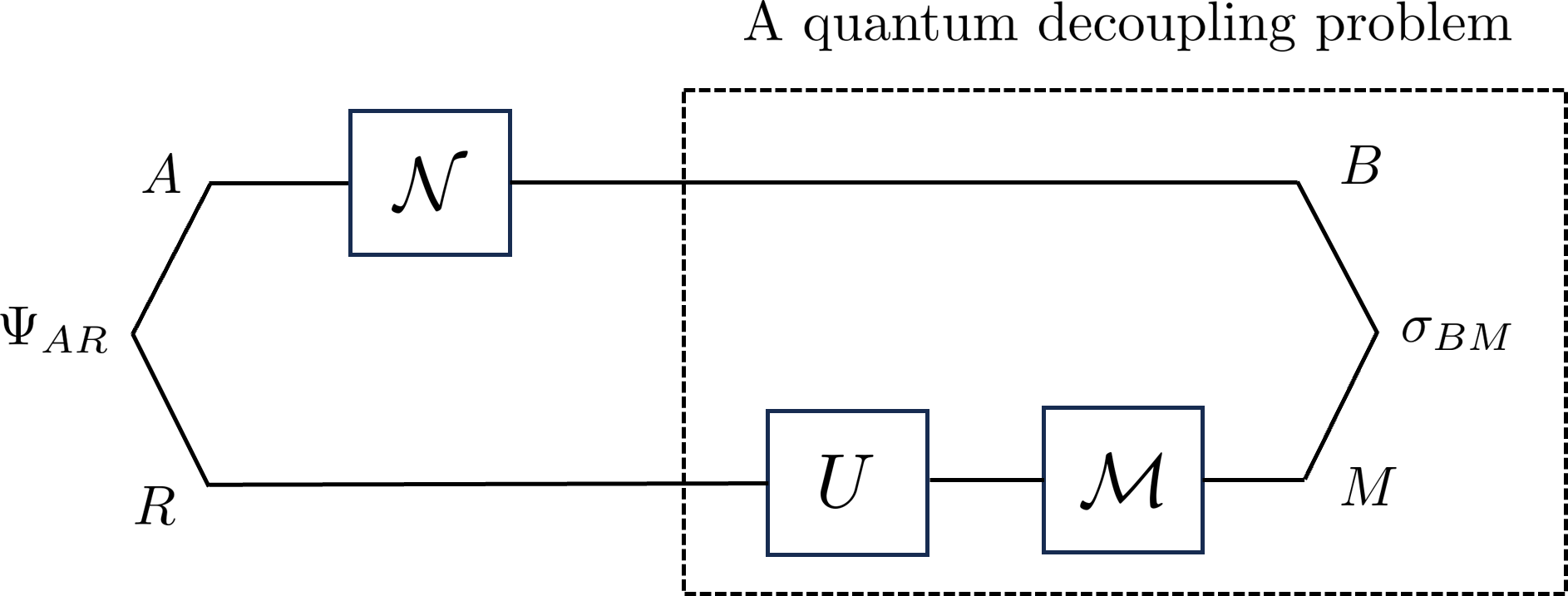}
  \caption{Random code construction in quantum soft covering}
  \label{fig-cover-decouple}
\end{figure}

\begin{itemize}[leftmargin=12.5pt]

\item [1)] Consider a state $\rho_A \in \mS(\rho_B, \mN)$. Let $|\mH_A| = d$. Consider the canonical purification of $\rho_A$:
$$ \ket{\psi} := \ket{\psi}_{AR} = (\sqrt{\rho_A}\otimes I_R) \ket{\Gamma}_{AR} = (\sqrt{\rho_A}\otimes I_R) \displaystyle{\sum_i} \ket{i}_A \ket{i}_R, $$
where $\{\ket{i}_A\}$ and $\{\ket{i}_R\}$ denote a basis for $\mH_A$ and $\mH_R$, respectively. In addition, we define $\Psi_{AR} := \ket{\psi}\bra{\psi}$.

\item [2)] The quantum channel $\mN$ acts on $\rho_A$ while leaving the reference system unchanged, yielding a bipartite output $\rho_{BR} = (\mN \otimes \id_R) (\Psi_{AR})$.

\item [3)] A random unitary operator $U$ with $|U| = d$ acts on system $R$. 

\item [4)] Perform a projection measurement $\mM = \{P_m\}_{m = 1,\dots,M}$ that consists of $M$ mutually orthogonal projectors which sum up to identity, i.e., $\sum_m P_m = I_M$. Each projector $P_m$ has the same rank $|P_m| = \Theta$, so $M\Theta = d$. The outcome $m$ is stored in a classical register $\ket{m}$. This measurement $\mM(\cdot)$ acting on any operator $\theta\in\mL(\mH_R)$ yields a classical state
\begin{equation}\label{proj}
    \mM(\theta)_M = \sum_m \tr[P_m \theta] \ket{m} \bra{m}.
\end{equation}
Note that for each outcome $m$, $(I\otimes P_m U)\ket{\psi}$ is a pure state with Schmidt rank no greater than $\Theta$, so the rank in system $A$ conditioned on $m$ is also no greater than $\Theta$.

\item [5)] After applying $\mM \circ U$ on system $R$, the state in system $BM$ becomes
\begin{align}
    \sigma_{BM} 
    &= (\id_B \otimes \mM)(I_B \otimes U) \rho_{BR} (I_B \otimes U^\dagger)  \label{sig_BM} \\
    &= \sum_m \tr_R \left[(I_B\otimes P_m U) \rho_{BR} (I_B\otimes U^\dagger P_m) \right] \otimes \ket{m} \bra{m}, \label{sig_BM-sum}
\end{align}
which is a normalized CQ state. Also, define
\begin{equation}\label{rho_BM}
    \rho_{BM} := \rho_B \otimes \pi_M 
    = \rho_B \otimes \frac{I_M}{M}.
\end{equation}

\end{itemize}

We have the following lemma.
\begin{lemma}\label{lem-exist sigB}
    If $\mathbb{E}_U D(\sigma_{BM}\|\rho_{BM}) \leq \epsilon$, then there exists a state $\sigma_A \in \mD(\mH_A)$ such that $D(\mN(\sigma_A)\|\rho_B) \leq \epsilon$ and $|\sigma_A| \leq \Theta$. Similarly, if \ $\mathbb{E}_U \left\|\sigma_{BM}-\rho_{BM}\right\|_1 \leq \epsilon$, then there exists a state $\sigma_A \in \mD(\mH_A)$ such that $\left\|\mN(\sigma_A)-\rho_B\right\|_1 \leq 2\epsilon$ and $|\sigma_A| \leq \Theta$.
\end{lemma}

\begin{proof}
In \eqref{sig_BM-sum} we define $p_m := \braket{m| \tr_B[\sigma_{BM}]|m}$ and can rewrite \eqref{sig_BM-sum} as
$$ \sigma_{BM} = \sum_m p_m \sigma_m \otimes \ket{m}\bra{m}, $$
which is a CQ state, where $\{p_m\}$ is a PMF and each $\sigma_m\in\mD(\mH_B)$ is a normalized state. 

We first show the relative entropy case. Note that
\begin{equation}\label{rela-BM}
    D(\sigma_{BM}\|\rho_{BM}) 
    = \sum_m p_m D(\sigma_m \| \rho_B) + D\left( p_1, \dots, p_M \bigg\| \frac{1}{M}, \dots, \frac{1}{M} \right).
\end{equation}
If $\mathbb{E}_U D(\sigma_{BM}\|\rho_{BM}) \leq \epsilon$, then there exists some unitary $U$ such that $D(\sigma_{BM}\|\rho_{BM}) \leq \epsilon$. Since all terms on the RHS of \eqref{rela-BM} are positive, there exists an outcome $m$ such that $D(\sigma_m \| \rho_B) \leq \epsilon$. 

For the trace distance case, similar arguments can be established by noting that
\begin{align*}
    \sum_m p_m \left\|\sigma_m -  \rho_B\right\|_1 
    &= \sum_m \left\|p_m \sigma_m - \frac{1}{M} \rho_B + \frac{1}{M} \rho_B - p_m \rho_B\right\|_1 \\
    &\leq \sum_m \left\|p_m \sigma_m - \frac{1}{M} \rho_B\right\|_1
        + \sum_m \left|\frac{1}{M} - p_m \right| \\
    &= \left\|\sigma_{BM}-\rho_{BM}\right\|_1 + \left\|\sigma_M-\rho_M\right\|_1 \\
    &\leq 2 \left\|\sigma_{BM}-\rho_{BM}\right\|_1 \leq 2\epsilon.
\end{align*}

In both cases, the obtained $\sigma_m$ can be written as $\sigma_m = \mN(\sigma_A)$ for some $\sigma_A \in \mD(\mH_A)$. From \eqref{sig_BM-sum} we can check that $|\sigma_A| \leq \Theta$ because $|P_m| \leq \Theta$. 
\end{proof}

\subsection{One-shot Quantum Soft Covering with Relative Entropy Criterion}\label{sec-cover-rela}

\begin{proof}[Proof of Theorem \ref{thm-cover-1shot}]
We first show achievability. Similar to the Section \ref{sec-decouple}, there are three steps. First, on the $R$ side we have a CPTP map $\mM$ defined by \eqref{proj}, which yields a CQ Choi operator
\begin{align}\label{tau_RM}
    \tau_{RM} 
    &= \frac{1}{d}(\id_B\otimes\mM) \left( \ket{\Gamma}_{RR'}\bra{\Gamma} \right)
    = \frac{1}{d} \sum_{ij} (\id_B\otimes\mM) \left(\ket{ii} \bra{jj} \right) \notag\\
    &= \frac{1}{d} \sum_m \sum_{ij} \braket{j|P_m|i} \ket{i}\bra{j} \otimes \ket{m} \bra{m}
    = \frac{1}{d} \sum_m P_m^T \otimes \ket{m}\bra{m}
\end{align}
with a marginal $\tau_M = \pi_M$. Hence, $\rho_{BM}$ in \eqref{rho_BM} can exactly be written as $\rho_{BM} = \rho_B \otimes \tau_M = \rho_B \otimes \mM(\pi_R)$. We may directly apply our decoupling result \eqref{decouple-1shot-Q2} (see Figure \ref{fig-cover-decouple}) and write
\begin{align}\label{cover-1shot-Q2}
    \mathbb{E}_U D(\sigma_{BM}\|\rho_{BM}) \leq
    \frac{\log e}{\Theta}\tlQ_2(\rho_{BR}\| \rho_B \otimes I_R),
\end{align}
where we have used the fact that $M\Theta = d$ and hence
$\tlQ_2(\tau_{RM}\| I_R \otimes \tau_M) = \frac{1}{\Theta}$. 

Next, we smooth the term $\tlQ_2(\rho_{BR}\| \rho_B \otimes I_R)$ in \eqref{cover-1shot-Q2} to a $-H_{\min}^\epsilon$ upper bound. By Definition \ref{def-Hmin^ep}, there exists a subnormalized state $\hrho_{BR} \in \mathscr{B}_\leq^\epsilon(\rho_{BR})$ such that 
\begin{equation}\label{find-hrhoBR}
    H_{\min}(R|B)_{\hrho_{BR}} = H_{\min}^\epsilon(R|B)_{\rho_{BR}}.
\end{equation}
Then we construct $\hsig_{BM}$ and $\hrho_{BM}$ corresponding to this $\hrho_{BR}$. We again employ the measurement $\mM$ and set $\hrho_{BM} := \hrho_B\otimes\tau_M$. Then we obtain $\hsig_{BM}$ analogous to \eqref{sig_BM}, by applying the operation $\mM\circ U$ on $\hrho_{BR}$:
$$ \hsig_{BM} = (\id_B \otimes \mM)(I_B \otimes U) \hrho_{BR} (I_B \otimes U^\dagger). $$
For this new $(\hsig_{BM}, \hrho_{BM})$ pair, \eqref{cover-1shot-Q2} gives
\begin{align}\label{Q:ED<Dmax^ep}
    \mathbb{E}D(\hsig_{BM} \| \hrho_{BM}) 
    &\leq \frac{\log e}{\Theta} \exp_2 \left\{ \log \tlQ_2(\hrho_{BR}\|\hrho_B \otimes I_R) \right\} \notag\\
    &\overset{a}{\leq} \frac{\log e}{\Theta} \exp_2 \left\{ \inf_{\xi_B \in \mD(\mH_B)} D_{\max}(\hrho_{BR}\| \xi_B \otimes I_R) \right\} \notag\\
    &\overset{b}{=} \frac{\log e}{\Theta} \exp_2 \left\{ -H_{\min}(R|B)_{\hrho_{BR}} \right\} \notag\\
    &\overset{c}{=} \frac{\log e}{\Theta} \exp_2 \left\{ -H_{\min}^\epsilon(R|B)_{\rho_{BR}}\right\},
\end{align}
where $(a)$ follows by Lemma \ref{lem-duality}, $(b)$ by Definition \ref{def-Hmin}, and $(c)$ by \eqref{find-hrhoBR}.
Then we need to evaluate trace distances $\mathbb{E}_U\left\|\sigma_{BM} - \hsig_{BM}\right\|_1$ and $\left\|\rho_{BM} - \hrho_{BM}\right\|_1$. The latter one should be straightforward: since $\hrho_{BR} \in \mathscr{B}_\leq^\epsilon(\rho_{BR})$, we have $\left\|\rho_{BR} - \hrho_{BR}\right\|_1 \leq 2\epsilon$ \cite[Lem.~6]{tomamichel2010duality} and thus 
\begin{equation}\label{Q:rho-rho<2ep}
    \left\|\rho_{BM} - \hrho_{BM}\right\|_1 = \left\|\rho_B - \hrho_B\right\|_1 \leq 2\epsilon,
\end{equation} 
as the data-processing inequality holds for unnormalized states \cite[Eq.~(4.1.7)]{khatri2020principles}. The other term $\mathbb{E}_U\left\|\sigma_{BM} - \hsig_{BM}\right\|_1$ can also be bounded using the data-processing inequality, since $\rho_{BR}$ and $\hrho_{BR}$ undergo the same channel $\mM\circ U$. Specifically, 
\begin{align}\label{Q:sig-sig<4ep}
    \mathbb{E}_U \left\|\sigma_{BM} - \hsig_{BM}\right\|_1
    & = \mathbb{E}_U \left\|(\id_B \otimes \mM)(I_B \otimes U) (\rho_{BR} - \hrho_{BR}) (I_B \otimes U^\dagger)\right\|_1 \notag \\
    &\leq \mathbb{E}_U \left\|(I_B \otimes U) (\rho_{BR} - \hrho_{BR}) (I_B \otimes U^\dagger)\right\|_1 \notag \\
    &= \left\|\rho_{BR} - \hrho_{BR}\right\|_1 \leq 2\epsilon.
\end{align}

Finally, we apply continuity. Inserting \eqref{Q:ED<Dmax^ep}, \eqref{Q:rho-rho<2ep}, and \eqref{Q:sig-sig<4ep} in Lemma \ref{lem-contin-rela} gives
\begin{equation}\label{ED-1shot-Q}
    \mathbb{E} D(\sigma_{BM} \| \rho_{BM}) 
    \leq \frac{\log e}{\Theta} \exp_2 \left\{ -H_{\min}^\epsilon(R|B)_{\rho_{BR}} \right\} + 4\epsilon \log M|\mH_B| + 2\ell(2\epsilon).
\end{equation}
Combining \eqref{ED-1shot-Q} and Lemma \ref{lem-exist sigB} gives our achievability result. Note that we have $M\leq d=|H_A|$. Also, we set $\eta = 2^{-H_{\min}^\epsilon(R|B) - \log\Theta}$ and further upper-bound it by $2\eta$ to account for rounding $\Theta$ down to the nearest smaller integer.

Then we prove the converse. Consider any $\rho'_A \in \mS_\delta(\rho_B, \mN)$ and a purification $\Psi'_{AR}$ of $\rho'_A$. We have 
\begin{align*}
    \log |\rho'_A|  
    \overset{a}{\geq} H_{\max}(R)_{\Psi'_{AR}} 
    \overset{b}{=} -H_{\min}(R|A)_{\Psi'_{AR}} 
    \overset{c}{\geq} -H_{\min}(R|B)_{\rho'_{BR}},
\end{align*}
where $(a)$ follows from the fact that R\'enyi entropy is monotone decreasing \cite[Cor.~4.13]{tomamichel-book}, $(b)$ follows from the duality between min- and max- entropies \cite[Prop.~5.7]{tomamichel-book}, $(c)$ follows from quantum data processing inequality \cite[Thm.~4.18]{tomamichel-book}, and $\rho'_{BR} = (\mN\otimes \id_R)\Psi'_{AR}$. The desired result follows by minimizing over all states in $\mS_{\delta}(\rho_B,\mN)$. 
\end{proof}

\subsection{Asymptotic Quantum Soft Covering with Relative Entropy Criterion} \label{sec-cover-asym}

\begin{proof}[Proof of Theorem \ref{thm-cover-asym}]
We first show the achievability. For any $\rho_A\in\mS(\rho_B,\mN)$, let $r > I(R\rangle B)_{\rho_{BR}}$, and $\zeta>0$ be arbitrary. Choose $\Theta=2^{n(r+\zeta)}$. By Definition \ref{def-Hmin},
\begin{equation}\label{Hmin<Dmax}
    -H_{\min}^\epsilon(R|B)_{\rho_{BR}} \leq D_{\max}^\epsilon(\rho_{BR} \| \rho_B \otimes I_R);
\end{equation}
thus in the $n$-shot case, \eqref{ED-1shot-Q} gives 
\begin{equation}\label{ED-nshot-Q}
    \mathbb{E}D(\sigma_{B^n M^n}\|\rho_{BM}^\tenn) 
    \leq \frac{\log e}{\Theta} \exp_2 \left\{D_{\max}^\epsilon(\rho_{BR}^\tenn \| \rho_B^\tenn \otimes I_R^\tenn) \right\} + 4\epsilon n \log|\mH_A||\mH_B| + 2\ell(2\epsilon).
\end{equation}
According to Lemma \ref{lem-Dmax-AEP}, the $D_{\max}^\epsilon$ term has the following AEP.
\begin{equation}\label{Hmin-AEP}
    \frac{1}{n} D_{\max}^\epsilon(\rho_{BR}^\tenn \| \rho_B^\tenn \otimes I_R^\tenn) \\
    = I(R\rangle B)_{\rho_{BR}} - \sqrt{\frac{ V(\rho_{BR}\| \rho_B \otimes I_R)} {n} } \ \Phi^{-1}(\epsilon^2) + O\left(\frac{\log n}{n}\right). 
\end{equation}
Note the $\epsilon n$ term on the RHS of \eqref{ED-nshot-Q}. To make \eqref{ED-nshot-Q} vanish at $n\to\infty$, we choose $\epsilon = \epsilon_n = \frac{1}{2e} n^{-\gamma}$ with $\gamma>1$. Under such a choice, 
\begin{equation}\label{Phi-to-0}
    \frac{\left|\Phi^{-1}(\epsilon_n^2)\right|}{\sqrt{n}} \overset{a}{\leq} \sqrt{\frac{-4\ln\epsilon_n}{n}} = \sqrt{\frac{4(1+\ln2)}{n} + 4\gamma \frac{\ln n}{n}} \to 0
\end{equation}
as $n\to\infty$ where in $(a)$ we have used \eqref{Phi<ln}.
Consequently, \eqref{ED-nshot-Q} vanishes and hence $r$ is achievable.

Then we show the converse, i.e., any achievable rate needs to be greater than $\inf_{\rho_A \in \mS(\rho_B,\mN)} I_c(\rho_A,\mN)$. Suppose $r$ is achievable. Then there is a code such that $D(\sigma_B\|\rho_B) \leq \delta$ where $\sigma_B$ is the encoded output state. By Lemma \ref{pinsker}, this code further satisfies $\left\|\sigma_B - \rho_B \right\|_1 \leq \sqrt{(2\ln2) \delta}$. Thus, we can directly apply the proof of converse in the trace-norm quantum covering problem, e.g., \cite[Thm.~3.5]{atif2024quantum}.
\end{proof}

\subsection{Error Exponent for Quantum Soft Covering}\label{sec-cover-errexp}

\begin{proof}[Proof of Theorem \ref{thm-cover-errexp-l1}]
Similar to Section \ref{sec-cover-rela}, we identify the random code construction after state $\rho_{BR}$ as a decoupling problem (see Figure \ref{fig-cover-decouple}). Therefore, we pick any $\rho_A\in\mS(\rho_B,\mN)$ and construct $\sigma_{BM},\rho_{BM}$ defined by \eqref{sig_BM} and \eqref{rho_BM}. 
It is shown in \cite[Thm.~2]{cheng2024joint} that 
\begin{equation}\label{EU-exp}
    \mathbb{E}_U \left\|\sigma_{BM} - \rho_{BM} \right\|_1
    \leq 2\cdot 3^{\frac{1-\alpha}{\alpha}}
        \exp_2\left\{ -\frac{\alpha-1}{\alpha}
        \left( \tilde{H}_\alpha^\uparrow(R|M)_{\tau_{RM}} + 
        \tilde{H}_\alpha^\uparrow(R|B)_{\rho_{BR}}\right) \right\}
\end{equation}
for any $\alpha\in[1,2]$. Here $\tau_{RM}$ is given by \eqref{tau_RM}.
By Definition \ref{def-Ha}, we have
$$ \tilde{H}_\alpha^\uparrow(R|M)_{\tau_{RM}}
    \geq - \tlD_\alpha(\tau_{RM}\|I_R\otimes\tau_M)
    = r $$
where the relation $M\Theta=d$ is used.
Also, $3^{\frac{1-\alpha}{\alpha}}$ in \eqref{EU-exp} contributes only as a constant in any $n$-shot case, so it suffices to write $3^{\frac{1-\alpha}{\alpha}} \leq 1$. Consequently, \eqref{EU-exp} has an upper bound
$$ \mathbb{E}_U \left\|\sigma_{BM} - \rho_{BM} \right\|_1
    \leq 2 \exp_2\left\{ -\frac{\alpha-1}{\alpha} \left( r + 
        \tilde{H}_\alpha^\uparrow(R|B)_{\rho_{BR}}\right) \right\}. $$
The desired conclusion thus follows immediately from Lemma \ref{lem-exist sigB} and optimizing over $\alpha\in[1,2]$. 
\end{proof}

\subsection{Second-order Rates for Quantum Soft Covering}\label{sec-cover-2nd}

\begin{proof}[Proof of Theorem \ref{thm-cover-2nd-l1}]
For simplicity, denote $I_c:=I_c(\rho_A,\mN)$ and $V:= V(\rho_A,\mN)$.  
Given a rate $r$, the optimizer $\alpha^*$ that yields $E(r)$ is solved via the stationary point equation
$$ 0 = \frac{d}{d\alpha}\left[
    \frac{\alpha-1}{\alpha}\left(r + \Ha\right)\right]_{\alpha=\alpha^*}
    = \frac{1}{(\alpha^*)^2}\left(r+\tlH_{\alpha^*}^\uparrow(R|B) \right)
    + \frac{\alpha^*-1}{\alpha^*} \frac{d}{d\alpha}\tHa \bigg|_{\alpha=\alpha^*}, $$
which yields an expression of $r$ in terms of $\alpha^*$. Therefore, We can relate $E(r)$ and $r$ parametrically as functions of $\alpha^*$ (for simplicity, we will use $\alpha$ to denote $\alpha^*$):
\begin{align*}
    r = r(\alpha) 
        &:= - \alpha(\alpha-1) \da \Ha - \Ha, \\
    E(r) = E_r(\alpha) 
        &:= \frac{\alpha-1}{\alpha}\left(r(\alpha) + \Ha\right) 
        = - (\alpha-1)^2 \da \Ha.
\end{align*}
Here we also write $\da:= \frac{\partial}{\partial\alpha}$, $\da^2:= \frac{\partial^2}{\partial\alpha^2}$, and $\Ha := \tHa$ to simplify notation. Then 
\begin{align*}
    \frac{dE(r)}{dr} 
    &= \frac{\da E_r(\alpha)}{\da r(\alpha)}
    = \frac{-2(\alpha-1)\da\Ha - (\alpha-1)^2\da^2\Ha}
        {-2\alpha\da\Ha - \alpha(\alpha-1)\da^2\Ha}
    = \frac{\alpha-1}{\alpha}, \\
    \frac{d^2E(r)}{dr^2}
    &= \frac{\da \frac{dE(r)}{dr}}{\da r(\alpha)}
    = \frac{1/\alpha^2}{-2\alpha\da\Ha - \alpha(\alpha-1)\da^2\Ha}
    = -\frac{1}{\alpha^3[2\da\Ha+(\alpha-1)\da^2\Ha]}.
\end{align*}
The optimizer $\alpha=1$ corresponds to $r = -\tlH_1^\uparrow(R|B) = -H(R|B) = I_c$. Expanding $E(r)$ at $r=I_c$ yields
\begin{align*}
    E(r)
    &= E(I_c) 
        + E'(I_c)(r-I_c)
        + \frac12 E''(I_c) (r-I_c)^2 
        + O\big((r-I_c)^3\big) \\
    &= E_r(\alpha=1) 
        + \frac{dE(r)}{dr}\bigg|_{\alpha=1}(r-I_c)
        + \frac12 \frac{d^2E(r)}{dr^2}\bigg|_{\alpha=1} (r-I_c)^2
        + O\big((r-I_c)^3\big) \\
    &= - \frac14 
        \left(\frac{dH_\alpha}{d\alpha}\bigg|_{\alpha=1}\right)^{-1} (r-I_c)^2 + O\big((r-I_c)^3\big) \\
    &= \frac{\log e}{2V} (r-I_c)^2 + O\big((r-I_c)^3\big),
\end{align*}
where the last equality follows from \cite[Prop.~11]{hayashi2016correlation}. Let $r$ approach $I_c$ from above at a rate
\begin{equation}\label{vn}
    r_n = I_c + \frac{v_n}{\sqrt{n}}
\end{equation}
with some sequence $v_n>0$. Then 
\begin{equation}\label{Er-n}
    E(r) = \frac{(\log e) v_n^2}{2Vn}
        + O\left(\frac{v_n^3}{n^{3/2}}\right).
\end{equation}
Let the trace-distance error be $\delta_n := \left\|\mN^\tenn(\sigma_{A^n}) - \rho_B^\tenn \right\|_1 $ with some encoded input state $\sigma_{A^n}\in\mD(\mH_A^\tenn)$.
From Theorem \ref{thm-cover-errexp-l1} we have
    $$ \frac{\delta_n}{4} \leq 2^{-nE(r)}
        = \exp_2\left\{ - \frac{\log e}{2V} v_n^2 + O\left(\frac{v_n^3}{\sqrt{n}}\right) \right\}. $$
Therefore, if $\lim_{n\to\infty}\delta_n \leq \delta$ for some fixed $\delta\in(0,1)$, we can set
$$ \exp_2\left\{ - \frac{\log e}{2V} v_n^2 \right\} \lesssim \frac{\delta}{4} $$
and hence 
$$ v_n^2 = \frac{2}{\log e} V \log\frac{4}{\delta} + o(1)
    = 2V \ln\frac{4}{\delta} + o(1). $$
The desired conclusion follows immediately from \eqref{vn}.
\end{proof}

\begin{proof}[Proof of Theorem \ref{thm-cover-2nd-D}]
Following notations in the previous proof, we write
$I_c:=I_c(\rho_A,\mN)$ and $V:= V(\rho_A,\mN)$. 
Let the relative-entropy covering error be $\delta_n := D\left(\mN^\tenn(\sigma_{A^n}) \| \rho_B^\tenn\right)$ for some encoded input state $\sigma_{A^n}\in\mD(\mH_A^\tenn)$.
Let $r$ approach $I_c$ via \eqref{vn} and thus the trace-distance error exponent scales as \eqref{Er-n}. According to the reverse Pinsker inequality in Lemma \ref{lem-revPinsker}, there exists a covering code such that
\begin{align*}
    \frac{\delta_n}{4} 
    \leq 2^{-nE(r)} n\left[ E(r) + \frac{\log|\mH_B|}{\eigrho}\right] 
    &\lesssim 2^{-\frac{\log e}{2V} v_n^2} 
        \left[ \frac{\log e}{2V} v_n^2
        + n \frac{\log|\mH_B|}{\eigrho} \right]. 
\end{align*}
Thus, to make $\lim_{n\to\infty}\delta_n \leq \delta$ for some given $\delta\in(0,1)$, we can take 
$$ 2^{-\frac{\log e}{2V} v_n^2} \cdot n \frac{\log|\mH_B|}{\eigrho} 
    \lesssim \frac{\delta}{8} $$
and hence 
$$ v_n = \sqrt{\frac{2V}{\log e} 
    \log\frac{8n\log|\mH_B|}{\eigrho\delta} + o(1)}
    = \sqrt{2V
    \ln\frac{8n\log|\mH_B|}{\eigrho\delta}} + o(1) $$
The desired conclusion follows immediately from \eqref{vn}.
\end{proof}

On the other hand, a second-order behavior for an $(n,\Theta,\delta)_D$ soft covering code can also be derived from the perspective of the one-shot achievability bound (Theorem \ref{thm-cover-1shot}). We present this derivation and compare it with Theorem \ref{thm-cover-2nd-D}.
For any $\rho_A\in\mS(\rho_B,\mN)$ we have $\rho_A^\tenn\in\mS(\rho_B^\tenn,\mN^\tenn)$, yielding a product bipartite output $\rho_{BR}^\tenn$. According to Theorem \ref{thm-cover-1shot}, for all $\epsilon\in(0,\frac{1}{2e})$ and $\eta\in(0,1)$, there exists an $(n,\Theta,\delta_n)_D$ quantum covering code with rate
\begin{align}\label{Theta-2nd}
    \log\Theta 
    &\leq \left[ -H_{\min}^\epsilon(R|B)_{\rho_{BR}^\tenn} - \log\eta \right]^+ \notag\\
    &\overset{a}{\leq} \left[ D_{\max}^\epsilon(\rho_{BR}^\tenn \| \rho_B^\tenn \otimes I_R^\tenn) - \log\eta \right]^+ \notag\\
    &\overset{b}{=} \left[ n I_c + \sqrt{n V} \left| \Phi^{-1}(\epsilon^2) \right| + O\left(\log n\right) - \log\eta \right]^+,
\end{align}
where $(a)$ follows from \eqref{Hmin<Dmax} and $(b)$ follows from \eqref{Hmin-AEP}. In addition, the error $\delta_n$ of this quantum covering code is bounded by
$$ \delta_n \leq 2(\log e)\eta + 4\epsilon n \log|\mH_A||\mH_B| + 2\ell(2\epsilon). $$
Picking some $\beta\in(0,1)$, we set $\eta = \frac{\beta}{2(\log e)}$ and $\epsilon = \frac{\beta}{n}$. 
By \eqref{Phi<ln} we have 
    $\left|\Phi^{-1}(\epsilon^2)\right| \lesssim \sqrt{4\ln(n/\beta)}$ 
and hence \eqref{Theta-2nd} reduces to
$$ \frac{1}{n}\log\Theta 
    \leq \left[ I_c + \sqrt{4V \ \frac{\ln (n/\beta)}{n}}  + O\left(\frac{\log n}{n}\right) \right]^+. $$
Furthermore, for large $n$, we can write $2\ell(2\epsilon)\leq\beta$ and then obtain $\delta_n \leq 2\beta + 4\beta \log|\mH_A||\mH_B|$. Thus, setting $\delta = 2\beta + 4\beta \log|\mH_A||\mH_B|$ can uniquely determine the constant $\beta = (2+4\log|\mH_A||\mH_B|)^{-1}\delta$. Thereby, we have shown the existence of an $(n,\Theta,\delta)_D$ whose rate scales as 
\begin{equation}\label{2nd-Dby1shot}
    \frac{1}{n}\log\Theta 
    = \left[ I_c(\rho_A,\mN) + \sqrt{4V(\rho_A,\mN) \ \frac{\ln (c'n)}{n}} + O\left(\frac{\log n}{n}\right) \right]^+
\end{equation}
with a constant $c' = (2+4\log|\mH_A||\mH_B|)\delta$. In \eqref{2nd-Dby1shot}, the dominant second-order term is $\sqrt{\frac{\ln n}{n}}$, which has coefficient $\sqrt{4V(\rho_A,\mN)}$, whereas the same term in \eqref{2nd-D} has a smaller coefficient $\sqrt{2V(\rho_A,\mN)}$. In conclusion, deriving the second-order term from the proposed error exponent $E(r)$ --- exactly as done in the proof of Theorem \ref{thm-cover-2nd-D} --- yields a tighter result than deriving it from the AEP of the one-shot achievability.

\section{Conclusion}
This work formulates the soft covering problems with relative entropy criterion for fully quantum channels and derives the achievable rate infimum. Our covering theorems achieve a closer approximation to the objective state than the usual trace-norm covering on account of the Pinsker inequality. To derive the one-shot bound, we provide a continuity and a quadratic bound of relative entropy in place of the triangle inequality and the quadratic bound of the trace norm, and apply the smoothing technique. These two bounds also yield a tighter result of the quantum decoupling theorem with relative entropy criterion. In addition, we establish achievable bounds on the error exponents and second-order rates for quantum soft covering.

\section*{Acknowledgements}
The authors would like to thank Touheed Anwar Atif and Andreas Winter for insightful discussions on this topic.


\appendices

\section{Continuity of Classical and Quantum Entropy for Subnormalized States}\label{app-conti}

The continuity of quantum entropy for normalized states is often given by the Fannes-Audenaert inequality \cite[Thms.~11.10.1 and 11.10.2]{wildebook}. Here, we present a similar continuity bound that applies to subnormalized states. We begin with the classical entropy and then extend to quantum entropy using the idea of pinching map. 

\begin{lemma}[Continuity of classical entropy] \label{lemma-contin-classical}
    Given two sequences $\{p_x\}_{x\in\mX}, \{q_x\}_{x\in\mX}$ such that $0 \leq p_x, q_x \leq 1$ for each $x\in\mX$ and that $$ T := \left\|p-q\right\|_1 = \sum_x \left| p_x - q_x \right| \leq \frac12, $$ we have
    $$ |H(p) - H(q)| \leq T \log|\mX| + \ell(T). $$
\end{lemma}

\begin{proof}
    We follow the strategy in \cite[Thm.~17.3.3]{coverbook} but extend it to unnormalized PMFs. Consider $v \leq \frac12$. For $0 \leq t \leq 1-v$, we have \cite[Thm.~17.3.3]{coverbook}
    $ |\ell(t) - \ell(t+v)| \leq -v\log v. $ 
    Define $r_x := |p_x - q_x|$. Then
    \begin{align*}
        |H(p) - H(q)| 
        &\leq \sum_x |-p_x \log p_x + q_x \log q_x| 
            \leq \sum_x -r_x \log r_x \\
        &= T \sum_x \frac{-r_x}{T} \left( \log\frac{r_x}{T} + \log T \right) 
        = T H\left(\frac{r_X}{T}\right) - T\log T \\
        &\overset{a}{\leq} T \log|\mX| - T\log T,
    \end{align*}
    where in $(a)$ we have used the fact that $r_X/T$ is a normalized and thus valid PMF.
\end{proof}

\begin{lemma}[Continuity of quantum entropy]\label{lem-contin-quantum}
    Given $\rho,\sigma \in \mD_\leq(\mH)$ such that $T := \left\|\rho - \sigma\right\|_1 \leq \frac1e$, we have
    $$ |H(\rho) - H(\sigma)| \leq T \log|\mH| + \ell(T). $$
\end{lemma}

\begin{proof}
    We follow the strategy in \cite[Thm.~11.10.1]{wildebook} but extend it to subnormalized states. Without loss of generality, assume that $H(\rho) \geq H(\sigma)$. Let the spectral decomposition of $\sigma$ be
    $$ \sigma = \sum_z \lambda_z \ket{z} \bra{z} $$
    and define the following ``pinching'' channel
    $$ \Delta(\rho) = \sum_z \ket{z} \braket{z|\rho|z} \bra{z}. $$
    Clearly $\Delta(\sigma) = \sigma$. Let $\tr[\rho] = 1/t$, so $t\rho$ is a normalized state. Then \cite[Cor.~11.9.3]{wildebook} gives $H(t\rho) \leq H(\Delta(t\rho)),$ where 
    \begin{align*}
        H(t\rho) = tH(\rho) - \log t, \quad
        H(\Delta(t\rho)) = tH(\Delta(\rho)) - \log t,
    \end{align*}
    so we have $H(\rho) \leq H(\Delta(\rho))$. 
    
    Define $T_\Delta := \left\| \Delta(\rho) - \Delta(\sigma)\right\|_1$. The data-processing inequality \cite[Eq.~(4.1.7)]{khatri2020principles} yields $T_\Delta \leq T \leq \frac1e$. We obtain
    \begin{align*}
        H(\rho) - H(\sigma) 
        \leq H(\Delta(\rho)) - H(\Delta(\sigma)) 
        \overset{a}{\leq} T_\Delta \log|\mH| + \ell(T_\Delta) 
        \overset{b}{\leq} T \log|\mH| + \ell(T),
    \end{align*}
    where $(a)$ is due to the fact that $\Delta(\rho)$ and $\Delta(\sigma)$ commute, so their quantum entropy and trace norm is the same as the classical entropy and norm, and hence we can apply Lemma \ref{lemma-contin-classical}; $(b)$ is because that $\ell(t) = -t\log t$ is monotone increasing on the interval $t\in[0,\frac1e]$.
\end{proof}

\begin{remark}
    Lemma \ref{lem-contin-quantum} cannot be generalized to any positive operators $\sigma,\rho\in\mP(\mH)$ because the classical sequences in Lemma \ref{lemma-contin-classical} are bounded from above.
\end{remark}

\section{Proofs of Some Lemmas}\label{app-proof}

\subsection{Proof of Lemma \ref{lem-contin-rela}}\label{app-proof-lem-contin}
\begin{proof}
Note that 
$$ D(\sigma\|\rho) - D(\hsig\|\hat{\rho}) = \left[H(\hsig) - H(\sigma)\right] + \tr\left[\hsig\log\hat{\rho} - \sigma\log\rho \right] $$
and thus
$$ \mathbb{E}D(\sigma\|\rho) - \mathbb{E}D(\hsig\|\hat{\rho}) = \mathbb{E} \left[H(\hsig) - H(\sigma)\right] + \left[H(\rho) - H(\hat{\rho})\right], $$
where the two terms are further bounded by, respectively, 
\begin{align*}
    \big|\mathbb{E}\left[H(\hsig) - H(\sigma)\right]\big| 
        &\leq \mathbb{E}|H(\hsig) - H(\sigma)| \\
        &\overset{a}{\leq} \mathbb{E} \left\|\sigma - \hsig\right\|_1 \log|\mH| + \mathbb{E} \ell (\left\|\sigma - \hsig\right\|_1) \\
        &\overset{b}{\leq} \mathbb{E} \left\|\sigma - \hsig\right\|_1 \log|\mH| + \ell(\mathbb{E} \left\|\sigma - \hsig\right\|_1) \\
        &\leq \epsilon_\sigma \log|\mH| + \ell(\epsilon_\sigma), \\
    \big|H(\rho) - H(\hat{\rho})\big|
        &\overset{c}{\leq} \epsilon_\rho \log|\mH| + \ell(\epsilon_\rho).
\end{align*}
Here $(a)$ and $(c)$ follow from Lemma \ref{lem-contin-quantum} in Appendix \ref{app-conti} and $(b)$ follows from the concavity of $\ell(\cdot)$.
\end{proof}

\subsection{Proof of Lemma \ref{lem-duality}}\label{app-proof-duality}

\begin{proof}
Let $\tr[\rho_{AB}] = \tr[\rho_B] = 1/t$. Then $ t \rho_{AB}, t \rho_B$ are normalized states. Note that
\begin{align*}
    \tlD_2 (t\rho_{AB} \|\tau_A \otimes t\rho_B)
        &= \log \tlQ_2 \left( t\rho_{AB} \|\tau_A \otimes t\rho_B \right) \\
        &= \log \tr \left[ t\rho_{AB} \left(\tau_A \otimes t\rho_B\right)^{-\frac12} t\rho_{AB} \left(\tau_A \otimes t\rho_B\right)^{-\frac12} \right] \\
        &= \log\tlQ_2 (\rho_{AB} \| \tau_A \otimes \rho_B) + \log t, \\
    \inf_{\xi_B \in \mD(\mH_B)} D_{\max}(t\rho_{AB}\|\tau_A \otimes \xi_B)
        &= \inf_{\xi_B \in \mD(\mH_B)} D_{\max}(t\rho_{AB}\|\tau_A \otimes \xi_B) \\
        &= \inf_{\xi_B \in \mD(\mH_B)} \inf \left\{\lambda \in \mathbb{R}: t\rho_{AB} \leq 2^\lambda (\tau_A \otimes \xi_B) \right\} \\
        &= \inf_{\xi_B \in \mD(\mH_B)} \inf \left\{\lambda \in \mathbb{R}: \rho_{AB} \leq 2^{\lambda - \log t} (\tau_A \otimes \xi_B) \right\} \\
        &\overset{a}{=} \inf_{\xi_B \in \mD(\mH_B)} 
            \inf \left\{\mu\in\mathbb{R}: \rho_{AB} \leq 2^\mu (\tau_A \otimes \xi_B) \right\} + \log t \\
        &= \inf_{\xi_B \in \mD(\mH_B)} D_{\max}(\rho_{AB}\|\tau_A \otimes \xi_B) + \log t,
\end{align*}
where $(a)$ follows by setting $\mu := \lambda - \log t$. To reach the desired conclusion, it suffices to show that 
\begin{equation}\label{dual}
    \tlD_2 (t\rho_{AB} \|\tau_A \otimes t\rho_B) 
    \leq \inf_{\xi_B \in \mD(\mH_B)} D_{\max}(t\rho_{AB}\|\tau_A \otimes \xi_B).
\end{equation}
Write $\sigma_{AB} := t\rho_{AB}\in\mD(\mH_A\otimes\mH_B)$ and consider any purification $\sigma_{ABC}$ of $\sigma_{AB}$.
Since $D_{\max} = \tlD_\infty$ (see Definition \ref{def-Dmax}), We have \cite[Lem.~6]{hayashi2016correlation}
\begin{align*}
    \tlD_2(t\rho_{AB}\|\tau_A \otimes \rho_B) 
        &= -\inf_{\sigma_C \in \mD(\mH_C)} D_{\frac12}(\sigma_{AC}\|\tau_A^{-1} \otimes \sigma_C), \\
    \inf_{\xi_B \in \mD(\mH_B)} D_{\max}(t\rho_{AB}\|\tau_A \otimes \xi_B) 
        &= -\inf_{\sigma_C \in \mD({H}_C)} \tlD_{\frac12}(\sigma_{AC}\|\tau_A^{-1} \otimes \sigma_C).
\end{align*}
Then \eqref{dual} is obvious since $D_{\frac12} \geq \tlD_{\frac12}$ \cite[Eq.~(4.5.45)]{khatri2020principles}.
\end{proof}

\subsection{Proof of Lemma \ref{lem-revPinsker}}\label{app-proof-revPinsker}
\begin{proof}
We begin with the definition of relative entropy.
\begin{align*}
    D(\sigma\|\rho) &= \tr\left[\sigma\log\sigma - \sigma\log\rho \right] \\
    &= \tr\left[\sigma\log\sigma - \rho\log\rho + \rho\log\rho - \sigma\log\rho \right] \\
    &= H(\rho) - H(\sigma) + \tr\left[(\rho-\sigma) \log\rho \right] \\
    &\overset{a}{\leq} |H(\rho) - H(\sigma)| 
        + \left\|\rho-\sigma\right\|_1 \left\|\log\rho\right\|_\infty \\
    &\overset{b}{\leq} \epsilon \log |\mH| + \ell(\epsilon) 
        + \epsilon \log \frac{1}{\eigrho} \\
    &= \epsilon \log \frac{|\mH|}{\eigrho} + \ell(\epsilon),
\end{align*}
where in $(a)$ we bound the second term $\tr\left[(\rho-\sigma) \log\rho \right]$ by H\"older's inequality and in $(b)$ we bound the first term $|H(\rho) - H(\sigma)|$ by Lemma \ref{lem-contin-quantum} in Appendix \ref{app-conti}.
\end{proof}

\bibliographystyle{IEEEtran}
\bibliography{ref}

@INPROCEEDINGS{he2024,
  author={He, Xingyi and Atif, Touheed Anwar and Pradhan, S. Sandeep},
  booktitle={2024 IEEE International Symposium on Information Theory (ISIT)}, 
  title={Quantum Soft Covering and Decoupling with Relative Entropy Criterion}, 
  year={2024},
  volume={},
  number={},
  pages={1426-1431}}

@article{bloch2013strong,
  title={Strong secrecy from channel resolvability},
  author={Bloch, Matthieu R and Laneman, J Nicholas},
  journal={IEEE Transactions on Information Theory},
  volume={59},
  number={12},
  pages={8077--8098},
  year={2013},
  publisher={IEEE}
}

@article{hayashi2006general,
  author={Hayashi, M.},
  journal={IEEE Transactions on Information Theory}, 
  title={General nonasymptotic and asymptotic formulas in channel resolvability and identification capacity and their application to the wiretap channel}, 
  year={2006},
  volume={52},
  number={4},
  pages={1562-1575}}

@article{csiszar2008axiomatic,
  title={Axiomatic characterizations of information measures},
  author={Csisz{\'a}r, Imre},
  journal={Entropy},
  volume={10},
  number={3},
  pages={261--273},
  year={2008},
  publisher={Molecular Diversity Preservation International}
}

@book{tomamichel-book,
  title={Quantum information processing with finite resources: mathematical foundations},
  author={Tomamichel, Marco},
  volume={5},
  year={2015},
  publisher={Springer}
}

@article{renner2008security,
  title={Security of quantum key distribution},
  author={Renner, Renato},
  journal={International Journal of Quantum Information},
  volume={6},
  number={01},
  pages={1--127},
  year={2008},
  publisher={World Scientific}
}

@article{wildebook,
  title={From classical to quantum Shannon theory},
  author={Wilde, Mark M},
  journal={arXiv preprint arXiv:1106.1445},
  year={2011}
}

@article{tomamichel2014relating,
  title={Relating different quantum generalizations of the conditional R{\'e}nyi entropy},
  author={Tomamichel, Marco and Berta, Mario and Hayashi, Masahito},
  journal={Journal of Mathematical Physics},
  volume={55},
  number={8},
  year={2014},
  publisher={AIP Publishing}
}

@article{atif2024quantum,
  title={Quantum soft-covering lemma with applications to rate-distortion coding, resolvability and identification via quantum channels},
  author={Atif, Touheed Anwar and Pradhan, S Sandeep and Winter, Andreas},
  journal={arXiv preprint arXiv:2306.12416},
  year={2023}
}

@article{dupuis2014one,
  title={One-shot decoupling},
  author={Dupuis, Fr{\'e}d{\'e}ric and Berta, Mario and Wullschleger, J{\"u}rg and Renner, Renato},
  journal={Communications in Mathematical Physics},
  volume={328},
  pages={251--284},
  year={2014},
  publisher={Springer}
}

@article{hayashi2016correlation,
  title={Correlation detection and an operational interpretation of the R{\'e}nyi mutual information},
  author={Hayashi, Masahito and Tomamichel, Marco},
  journal={Journal of Mathematical Physics},
  volume={57},
  number={10},
  year={2016},
  publisher={AIP Publishing}
}

@article{anshu2019convex,
  title={Convex-split and hypothesis testing approach to one-shot quantum measurement compression and randomness extraction},
  author={Anshu, Anurag and Jain, Rahul and Warsi, Naqueeb Ahmad},
  journal={IEEE Transactions on Information Theory},
  volume={65},
  number={9},
  pages={5905--5924},
  year={2019},
  publisher={IEEE}
}

@article{tomamichel2013hierarchy,
  title={A hierarchy of information quantities for finite block length analysis of quantum tasks},
  author={Tomamichel, Marco and Hayashi, Masahito},
  journal={IEEE Transactions on Information Theory},
  volume={59},
  number={11},
  pages={7693--7710},
  year={2013},
  publisher={IEEE}
}

@book{coverbook,
  title={Elements of information theory},
  author={Cover, Thomas M},
  year={1999},
  publisher={John Wiley \& Sons}
}

@article{tomamichel2010duality,
  author={Tomamichel, Marco and Colbeck, Roger and Renner, Renato},
  journal={IEEE Transactions on Information Theory}, 
  title={Duality Between Smooth Min- and Max-Entropies}, 
  year={2010},
  volume={56},
  number={9},
  pages={4674-4681},
  doi={10.1109/TIT.2010.2054130}}

@article{khatri2020principles,
  title={Principles of quantum communication theory: A modern approach},
  author={Khatri, Sumeet and Wilde, Mark M},
  journal={arXiv preprint arXiv:2011.04672},
  year={2020}
}

@article{cheng2023error,
  title={Error exponent and strong converse for quantum soft covering},
  author={Cheng, Hao-Chung and Gao, Li},
  journal={IEEE Transactions on Information Theory},
  volume={70},
  number={5},
  pages={3499--3511},
  year={2023},
  publisher={IEEE}
}

@article{shen2024optimal,
  title={Optimal second-order rates for quantum soft covering and privacy amplification},
  author={Shen, Yu-Chen and Gao, Li and Cheng, Hao-Chung},
  journal={IEEE Transactions on Information Theory},
  year={2024},
  publisher={IEEE}
}

@article{han2002approximation,
  author={Han, T.S. and Verdu, S.},
  journal={IEEE Transactions on Information Theory}, 
  title={Approximation theory of output statistics}, 
  year={1993},
  volume={39},
  number={3},
  pages={752-772},
  doi={10.1109/18.256486}}

@book{cuff2009communication,
  title={Communication in networks for coordinating behavior},
  author={Cuff, Paul W},
  year={2009},
  publisher={Stanford University}
}

@article{wyner1975common,
  title={The common information of two dependent random variables},
  author={Wyner, Aaron},
  journal={IEEE Transactions on Information Theory},
  volume={21},
  number={2},
  pages={163--179},
  year={1975},
  publisher={IEEE}
}

@article{cuff2013distributed,
  title={Distributed channel synthesis},
  author={Cuff, Paul},
  journal={IEEE Transactions on Information Theory},
  volume={59},
  number={11},
  pages={7071--7096},
  year={2013},
  publisher={IEEE}
}

@inproceedings{cuff2016soft,
  title={Soft covering with high probability},
  author={Cuff, Paul},
  booktitle={2016 IEEE International Symposium on Information Theory (ISIT)},
  pages={2963--2967},
  year={2016},
  organization={IEEE}
}

@ARTICLE{yagli2019exact,
  author={Yagli, Semih and Cuff, Paul},
  journal={IEEE Transactions on Information Theory}, 
  title={Exact Exponent for Soft Covering}, 
  year={2019},
  volume={65},
  number={10},
  pages={6234-6262},
  doi={10.1109/TIT.2019.2917182}}

@article{wilde2012information,
  title={The information-theoretic costs of simulating quantum measurements},
  author={Wilde, Mark M and Hayden, Patrick and Buscemi, Francesco and Hsieh, Min-Hsiu},
  journal={Journal of Physics A: Mathematical and Theoretical},
  volume={45},
  number={45},
  pages={453001},
  year={2012},
  publisher={IOP Publishing}
}

@article{dupuis-thesis,
  title={The decoupling approach to quantum information theory},
  author={Dupuis, Fr{\'e}d{\'e}ric},
  journal={arXiv preprint arXiv:1004.1641},
  year={2010}
}

@article{winter2004extrinsic,
  title={‘‘Extrinsic’’and ‘‘Intrinsic’’Data in Quantum Measurements: Asymptotic Convex Decomposition of Positive Operator Valued Measures},
  author={Winter, Andreas},
  journal={Communications in mathematical physics},
  volume={244},
  pages={157--185},
  year={2004},
  publisher={Springer}
}

@article{ahlswede2002strong,
  title={Strong converse for identification via quantum channels},
  author={Ahlswede, Rudolf and Winter, Andreas},
  journal={IEEE Transactions on Information Theory},
  volume={48},
  number={3},
  pages={569--579},
  year={2002},
  publisher={IEEE}
}

@article{devetak2003classical,
  title={Classical data compression with quantum side information},
  author={Devetak, Igor and Winter, Andreas},
  journal={Physical Review A},
  volume={68},
  number={4},
  pages={042301},
  year={2003},
  publisher={APS}
}

@article{devetak2005distillation,
  title={Distillation of secret key and entanglement from quantum states},
  author={Devetak, Igor and Winter, Andreas},
  journal={Proceedings of the Royal Society A: Mathematical, Physical and engineering sciences},
  volume={461},
  number={2053},
  pages={207--235},
  year={2005},
  publisher={The Royal Society}
}

@inproceedings{winter2005secret,
  title={Secret, public and quantum correlation cost of triples of random variables},
  author={Winter, Andreas},
  booktitle={Proceedings. International Symposium on Information Theory, 2005. ISIT 2005.},
  pages={2270--2274},
  year={2005},
  organization={IEEE}
}

@article{cai2004quantum,
  title={Quantum privacy and quantum wiretap channels},
  author={Cai, Ning and Winter, Andreas and Yeung, Raymond W},
  journal={problems of information transmission},
  volume={40},
  pages={318--336},
  year={2004},
  publisher={Springer}
}

@article{bennett2014quantum,
  title={The quantum reverse Shannon theorem and resource tradeoffs for simulating quantum channels},
  author={Bennett, Charles H and Devetak, Igor and Harrow, Aram W and Shor, Peter W and Winter, Andreas},
  journal={IEEE Transactions on Information Theory},
  volume={60},
  number={5},
  pages={2926--2959},
  year={2014},
  publisher={IEEE}
}

@article{luo2009channel,
  title={Channel simulation with quantum side information},
  author={Luo, Zhicheng and Devetak, Igor},
  journal={IEEE Transactions on Information Theory},
  volume={55},
  number={3},
  pages={1331--1342},
  year={2009},
  publisher={IEEE}
}

@article{radhakrishnan2017one,
  title={One-shot private classical capacity of quantum wiretap channel: Based on one-shot quantum covering lemma},
  author={Radhakrishnan, Jaikumar and Sen, Pranab and Warsi, Naqueeb Ahmad},
  journal={arXiv preprint arXiv:1703.01932},
  year={2017}
}

@article{2001compression,
  title={Compression of quantum-measurement operations},
  author={Winter, Andreas and Massar, Serge},
  journal={Physical Review A},
  volume={64},
  number={1},
  pages={012311},
  year={2001},
  publisher={APS}
}

@article{massar2000amount,
  title={Amount of information obtained by a quantum measurement},
  author={Massar, Serge and Popescu, Sandu},
  journal={Physical Review A},
  volume={61},
  number={6},
  pages={062303},
  year={2000},
  publisher={APS}
}

@ARTICLE{parizi2016exact,
  author={Bastani Parizi, Mani and Telatar, Emre and Merhav, Neri},
  journal={IEEE Transactions on Information Theory}, 
  title={Exact Random Coding Secrecy Exponents for the Wiretap Channel}, 
  year={2017},
  volume={63},
  number={1},
  pages={509-531}}

@ARTICLE{yu2018renyi,
  author={Yu, Lei and Tan, Vincent Y. F.},
  journal={IEEE Transactions on Information Theory}, 
  title={Rényi Resolvability and Its Applications to the Wiretap Channel}, 
  year={2019},
  volume={65},
  number={3},
  pages={1862-1897}}

@ARTICLE{hayashi2012quantumwiretap,
  author={Hayashi, Masahito},
  journal={IEEE Transactions on Information Theory}, 
  title={Quantum Wiretap Channel With Non-Uniform Random Number and Its Exponent and Equivocation Rate of Leaked Information}, 
  year={2015},
  volume={61},
  number={10},
  pages={5595-5622}}

@book{hayashi-book,
  title={Quantum information theory},
  author={Hayashi, Masahito},
  year={2016},
  publisher={Springer}
}

@article{hayashi2006formula,
  title={General formulas for fixed-length quantum entanglement concentration},
  author={Hayashi, Masahito},
  journal={IEEE transactions on information theory},
  volume={52},
  number={5},
  pages={1904--1921},
  year={2006},
  publisher={IEEE}
}

@article{li2024reliability,
  title={Reliability function of quantum information decoupling via the sandwiched R{\'e}nyi divergence},
  author={Li, Ke and Yao, Yongsheng},
  journal={Communications in Mathematical Physics},
  volume={405},
  number={7},
  pages={160},
  year={2024},
  publisher={Springer}
}

@article{dupuis2023privacy,
  title={Privacy amplification and decoupling without smoothing},
  author={Dupuis, Fr{\'e}d{\'e}ric},
  journal={IEEE Transactions on Information Theory},
  volume={69},
  number={12},
  pages={7784--7792},
  year={2023},
  publisher={IEEE}
}

@article{horodecki2007quantum,
  title={Quantum state merging and negative information},
  author={Horodecki, Micha{\l} and Oppenheim, Jonathan and Winter, Andreas},
  journal={Communications in Mathematical Physics},
  volume={269},
  number={1},
  pages={107--136},
  year={2007},
  publisher={Springer}
}

@article{abeyesinghe,
    author = {Abeyesinghe, Anura and Devetak, Igor and Hayden, Patrick and Winter, Andreas},
    title = {The mother of all protocols: restructuring quantum information’s family tree},
    journal = {Proceedings of the Royal Society A: Mathematical, Physical and Engineering Sciences},
    volume = {465},
    number = {2108},
    pages = {2537-2563},
    year = {2009},
    month = {06}
}

@article{berta2011quantum,
  author  = {Berta, Mario and Christandl, Matthias and Renner, Renato},
  title   = {The Quantum Reverse Shannon Theorem Based on One-Shot Information Theory},
  journal = {Communications in Mathematical Physics},
  volume  = {306},
  number  = {3},
  pages   = {579--615},
  year    = {2011}
}

@article{majenz2017catalytic,
  title = {Catalytic Decoupling of Quantum Information},
  author = {Majenz, Christian and Berta, Mario and Dupuis, Fr\'ed\'eric and Renner, Renato and Christandl, Matthias},
  journal = {Phys. Rev. Lett.},
  volume = {118},
  issue = {8},
  pages = {080503},
  numpages = {6},
  year = {2017},
  month = {Feb},
  publisher = {American Physical Society}
}

@article{hayden2008decoupling,
    author = {Hayden, Patrick and Horodecki, Micha\l{} and Winter, Andreas and Yard, Jon},
    title = {A Decoupling Approach to the Quantum Capacity},
    journal = {Open Systems \& Information Dynamics},
    volume = {15},
    number = {01},
    pages = {7-19},
    year = {2008}
}

@article{wakakuwa2021oneshot,
  author  = {Wakakuwa, Etsuko and Nakata, Yoshifumi},
  title   = {One-Shot Randomized and Nonrandomized Partial Decoupling},
  journal = {Communications in Mathematical Physics},
  volume  = {386},
  number  = {1},
  pages   = {589--649},
  year    = {2021}
}

@book{vershynin-book,
  title={High-dimensional probability: An introduction with applications in data science},
  author={Vershynin, Roman},
  volume={47},
  year={2018},
  publisher={Cambridge university press}
}

@INPROCEEDINGS{yassaee2019almost,
  author={Yassaee, Mohammad H.},
  booktitle={Proc. 2019 IEEE International Symposium on Information Theory (ISIT)}, 
  title={{Almost Exact Analysis of Soft Covering Lemma via Large Deviation}}, 
  year={2019},
  pages={1387-1391},
  doi={10.1109/ISIT.2019.8849341}
}

@ARTICLE{yu2025renyi,
  author={Yu, Lei},
  journal={IEEE Transactions on Information Theory}, 
  title={Rényi Resolvability, Noise Stability, and Anti-Contractivity}, 
  year={2025},
  volume={71},
  number={8},
  pages={5836-5867}
}

@article{li2025two,
  title={Two-Parameter R{\'e}nyi Information Quantities with Applications to Privacy Amplification and Soft Covering},
  author={Li, Shi-Bing and Li, Ke and Yu, Lei},
  journal={arXiv preprint arXiv:2511.02297},
  year={2025}
}

@ARTICLE{hayashi2025resolvability,
  author={Hayashi, Masahito and Cheng, Hao-Chung and Gao, Li},
  journal={IEEE Transactions on Information Theory}, 
  title={Resolvability of Classical-Quantum Channels}, 
  year={2025},
  volume={71},
  number={8},
  pages={6061-6074},
}

@article{matsuura2025universal,
  title={Universal classical-quantum channel resolvability and private channel coding},
  author={Matsuura, Takaya and Hayashi, Masahito and Hsieh, Min-Hsiu},
  journal={arXiv preprint arXiv:2510.02883},
  year={2025}
}

@article{takahashi2026classical,
  title={Classical-Quantum Channel Resolvability Using Matrix Multiplicative Weight Update Algorithm},
  author={Takahashi, Koki and Watanabe, Shun},
  journal={arXiv preprint arXiv:2601.12230},
  year={2026}
}

@inproceedings{watanabe2014strong,
  title={Strong converse and second-order asymptotics of channel resolvability},
  author={Watanabe, Shun and Hayashi, Masahito},
  booktitle={2014 IEEE International Symposium on Information Theory},
  pages={1882--1886},
  year={2014},
  organization={IEEE}
}

@article{liu2016e_,
  title={{$E_\gamma$-Resolvability}},
  author={Liu, Jingbo and Cuff, Paul and Verd{\'u}, Sergio},
  journal={IEEE Transactions on Information Theory},
  volume={63},
  number={5},
  pages={2629--2658},
  year={2016},
  publisher={IEEE}
}

@article{cheng2024joint,
  title={Joint state-channel decoupling and one-shot quantum coding theorem},
  author={Cheng, Hao-Chung and Dupuis, Fr{\'e}d{\'e}ric and Gao, Li},
  journal={arXiv preprint arXiv:2409.15149},
  year={2024}
}

@article{cheng2025sharp,
  title={Sharp estimates of quantum covering problems via a novel trace inequality},
  author={Cheng, Hao-Chung and Gao, Li and Hirche, Christoph and Huang, Hao-Wei and Liu, Po-Chieh},
  journal={arXiv preprint arXiv:2507.07961},
  year={2025}
}

@article{berta2009single,
  title={Single-shot quantum state merging},
  author={Berta, Mario},
  journal={arXiv preprint arXiv:0912.4495},
  year={2009}
}

@article{datta2012one,
  title={One-shot entanglement-assisted quantum and classical communication},
  author={Datta, Nilanjana and Hsieh, Min-Hsiu},
  journal={IEEE Transactions on Information Theory},
  volume={59},
  number={3},
  pages={1929--1939},
  year={2012},
  publisher={IEEE}
}

@article{datta2014limit,
  title={A limit of the quantum R{\'e}nyi divergence},
  author={Datta, Nilanjana and Leditzky, Felix},
  journal={Journal of Physics A: Mathematical and Theoretical},
  volume={47},
  number={4},
  pages={045304},
  year={2014},
  publisher={IOP Publishing}
}

@article{datta2013smooth,
  title={A smooth entropy approach to quantum hypothesis testing and the classical capacity of quantum channels},
  author={Datta, Nilanjana and Mosonyi, Milan and Hsieh, Min-Hsiu and Brandao, Fernando GSL},
  journal={IEEE transactions on information theory},
  volume={59},
  number={12},
  pages={8014--8026},
  year={2013},
  publisher={IEEE}
}

@article{yu2019simulation,
  title={Simulation of random variables under R{\'e}nyi divergence measures of all orders},
  author={Yu, Lei and Tan, Vincent YF},
  journal={IEEE Transactions on Information Theory},
  volume={65},
  number={6},
  pages={3349--3383},
  year={2019},
  publisher={IEEE}
}

@article{yu2018wyner,
  title={Wyner’s common information under R{\'e}nyi divergence measures},
  author={Yu, Lei and Tan, Vincent YF},
  journal={IEEE Transactions on Information Theory},
  volume={64},
  number={5},
  pages={3616--3632},
  year={2018},
  publisher={IEEE}
}

\end{document}